\documentclass[12pt]{article}
\usepackage{amsfonts}
\usepackage{graphicx}
\usepackage{setspace}
\usepackage{natbib}
\usepackage{subfigure}
\usepackage{xfrac}
\usepackage[export]{adjustbox}
\usepackage{amssymb}
\usepackage{amsmath}

\newtheorem{prop}{Proposition}[section]
\newtheorem{theorem}[prop]{Theorem}
\newtheorem{lemma}[prop]{Lemma}

\newenvironment{proof}[1][Proof]{\noindent\textbf{#1.} }{\ \rule{0.5em}{0.5em}}

\begin{document}

\title{Model Based Bootstrap Methods for Interval Censored Data}
\author{Bodhisattva Sen and Gongjun Xu\\
Columbia University and University of Minnesota}
\maketitle

\begin{abstract}
We investigate the performance of model based bootstrap methods for constructing point-wise confidence intervals around the survival function with interval censored data. We show that bootstrapping from the nonparametric maximum likelihood estimator of the survival function is inconsistent for both the current status and case 2 interval censoring models. A model based smoothed bootstrap procedure is proposed and shown to be consistent. In addition, simulation studies are conducted to illustrate the (in)-consistency of the bootstrap methods. Our conclusions in the interval censoring model would extend  more generally to estimators in regression models that exhibit non-standard  rates of convergence. 
\end{abstract}

\section{Introduction}\label{Intro}
In recent years there has been considerable research on the analysis of interval censored data. Such data arise extensively in epidemiological studies and clinical trials, especially in large-scale panel studies where the event of interest, which is typically an infection with a disease or some other failure (like organ failure), is not observed exactly but is only known to happen between two consecutive examination times. In particular, large-scale HIV/AIDS studies typically yield various types of interval censored data where interest centers on the distribution of time to HIV infection, but the exact time of infection is only known to lie between two consecutive followups at the clinic. 

For general interval censored data, often called {\it mixed case} interval censoring, an individual is checked at several time points and the status of the individual is ascertained: 1 if the infection/failure has occurred by the time he/she is checked and 0 otherwise.  Let $X$ be the unobserved time of onset of some disease, having distribution function  $F$, and let $T_1\leq T_2\leq \cdots\leq T_K$ be the $K$ observation times.  Here $X$ and $(T_1,\ldots, T_K)$ are assumed to be independent and $K$ can be random. We observe $(T_1,\ldots, T_K,\Delta_1,\ldots, \Delta_K)$ where $\Delta_k = \mathbf{1}_{T_{k-1}< X \le T_k},  k=1,\ldots,K,$ with $T_0=0$ and $\mathbf{1}$ denoting the indicator function.
Our data consist of $n$ independent and identically distributed copies of $(T_1,\ldots, T_K,\Delta_1,\ldots, \Delta_K)$.  We are interested in making inference about the value of $F$ at a pre-specified location $t_0 >0$, assumed to be in the interior of the support~of~$F$.

When the observation number $K \equiv 1$, we say that    we have {\it case 1} interval censoring or {\it current status data}.  In this case, 
our observations are $(T_i,\Delta_i)$ with $\Delta_i = \mathbf{1}_{X_i \le T_i}$, $i=1,\ldots,n$.
The nonparametric maximum likelihood estimator (NPMLE) $\tilde F_n$ of $F$ maximizes the 
log-likelihood function
\begin{equation}\label{npmle}
 \mathbb{F} \mapsto {\sum}_{i=1}^n \{\Delta_i \log \mathbb{F}(T_i) + (1 - \Delta_i) \log (1 - \mathbb{F}(T_i))\}
 \end{equation}
over all distribution functions $\mathbb{F}$ and it can be characterized as the left derivative of the greatest convex minorant of the cumulative sum diagram of the data; see page 41 of~\cite{GW92}.  Let $G$ be the distribution function of $T$ and assume that $F$ and $G$ are continuously differentiable at $t_0$ with derivatives $f(t_0) > 0$ and $g(t_0) > 0$, respectively. Under these assumptions, it is well known that
\begin{equation}\label{eq:limF_n1}
\gamma_n:=  n^{1/3} \{\tilde F_n(t_0) - F(t_0)\} \stackrel{}{\rightarrow} \kappa \mathbb{C},
\end{equation}
in distribution, 
where $\kappa = [4 F(t_0)\{1-F(t_0)\}f(t_0)/g(t_0)]^{1/3}$, $\mathbb{C} = \arg \min_{h \in \mathbb{R}} \{\mathbb{Z}(h) + h^2\}$, and $\mathbb{Z}$ is a standard two-sided Brownian motion process, originating from $0$. 

In the general mixed case interval censoring model, the limiting distribution of the NPMLE is unknown. In fact, in the literature only very limited theoretical results are available on the NPMLE.  \cite{GW92} discussed the asymptotics of the behavior of the NPMLE in a particular version of the case 2 censoring model ($K=2$); \cite{wellner1995interval} studied the consistency of the NPMLE where each subject gets exactly $k$ examination times;  \cite{van2000preservation} proved the consistency of the NPMLE of the mixed case interval censoring in the Hellinger distance; see also  \cite{schick2000consistency} and \cite{song2004estimation}.

We are interested in constructing a pointwise confidence interval  for $F$ at $t_0$ in the general mixed case  censoring model. 
In the literature, very few results exist that address the construction of pointwise confidence intervals \citep{song2004estimation,sen2007pseudolikelihood}. Even in the current status model, where we know the limiting distribution of the NPMLE, to construct a confidence interval for $F(t_0)$ we need to estimate the nuisance parameter $\kappa$, which is indeed quite difficult -- it involves estimation of the derivative of $F$ and that of the distribution of $T$. For the current status model, there exist a few methods that can be used for constructing confidence intervals for $F(t_0)$: The $m$-out-of-$n$ bootstrap method and subsampling are known to be consistent in this setting \citep*{PRW99,LP06}. However, both methods require the choice of a {\it block} size. In practice, the choice of this tuning parameter is quite tricky and the confidence intervals vary drastically with different choices of the block size. The estimation of the nuisance parameter $\kappa$ can be avoided by using the likelihood-ratio test of \cite{BW01}. Recently, \cite*{groeneboom2010maximum} proposed estimates of $F(t_0)$ based on smoothed likelihood function and smoothed NPMLE. However, the limiting distributions depend on the derivative of the density function.

In this paper we consider bootstrap methods for constructing confidence intervals for $F(t_0)$ and investigate the (in)-consistency and performance of two model-based bootstrap procedures that are based on the NPMLE of $F$, in the general framework of mixed case interval censoring. Bootstrap intervals avoid the problem of estimating nuisance parameters and are generally reliable in problems with $\surd{n}$ convergence rates. See \cite{BF81}, \cite{S81}, and \cite{ST95} and  references therein. 

In regression models, there are two main bootstrapping strategies: ``bootstrapping pairs'' and ``bootstrapping residuals''  \citep[see e.g., page 113 of][]{ET93}. \cite{AH05} considered ``bootstrapping pairs'', i.e., bootstrapping from the empirical distribution function of the data, and showed that the procedure is {\it inconsistent} for the current status model and also other cube-root convergent estimators. In ``bootstrapping residuals'' one fixes (conditions on) the predictor values and generates the response according to the estimated regression model using bootstrapped residuals. In a binary regression problem, as in the current status model, this corresponds to generating the responses as independent Bernoulli random variables with success probability obtained from the fitted regression model. 

In this paper we focus on the ``bootstrapping residuals'' procedure. In particular, for the mixed-case interval censoring model, conditional on an individual's observation times $T_1,\ldots, T_K$, we generate bootstrap sample $(\Delta^*_1,\ldots,\Delta^*_K, 1-\sum_{k=1}^K\Delta^*_k)$ following a multinomial distribution with $n=1$ and $p_k= \hat F_n(T_k)-\hat F_n(T_{k-1})$, $k=1,\ldots,K+1,$ i.e., $$ \vspace{-0.02in}( \Delta^*_1,\ldots,\Delta^*_K, 1-{\sum}_{k=1}^{K}\Delta^*_k ) \sim \mbox{Multinomial}(1,\{\hat F_n(T_k)-\hat F_n(T_{k-1})\}_{i=1}^{K+1}),$$ where $\hat F_n$ is an estimator of $F$ and $\hat F_n(T_0)=0,\hat F_n(T_{K+1})=1$. We call this a model-based bootstrap scheme, as it uses the inherent features of the model. We study the behavior of the bootstrap method when $\hat F_n = \tilde F_n$, the NPMLE of $F$, and $\hat F_n$ is a smooth estimator of $F$. Specifically, in Section~\ref{SuffCond} we state a general bootstrap convergence result for the current status model which provides sufficient conditions for any bootstrap scheme to be consistent. In Section~\ref{InconsB} we illustrate, both theoretically and through simulation, the inconsistency of the NPMLE bootstrap method. 
The failure of the NPMLE bootstrap is mostly due to the non-differentiability  of $\tilde F_n$. 
On the other hand, the smoothed NPMLE is differentiable and successfully mimics the local behavior of the true distribution function $F$  at the location of interest, i.e., $t_0$. As a result, the method yields asymptotically valid confidence intervals; see Section~\ref{ConsSB} where we prove the consistency of the smoothed bootstrap procedure, again in the current status model. The smoothed bootstrap procedure requires the choice of a smoothing bandwidth and we discuss this problem of bandwidth selection in Section~\ref{TunPara}.

Next, in Section~\ref{case2}, we study  the case 2 interval censoring model, i.e., when $K \equiv 2$. Even in this case, the distribution of the NPMLE is not completely known, although conjectures and partial results exist. \cite{groeneboom1991nonparametric} studied a one-step estimate $F_n^{(1)}$, obtained at the first step of the iterative convex minorant algorithm  \citep[see][]{GW92} and conjectured that  $F^{(1)}_n$ is asymptotically equivalent to the NPMLE. This conjecture is called {\it the working hypothesis} in this paper and is still unproved. We assume that this conjecture holds and focus on bootstrapping the distribution of the one-step estimator. We show the inconsistency of bootstrapping from the NPMLE and  the consistency of the smoothed bootstrap method.

In general mixed case interval censoring, \cite{sen2007pseudolikelihood}  introduced a pseudolikelihood method for estimating $F(t_0)$. However, the pseudolikelihood does not use the full information in the data and may not be as efficient as the NPMLE \citep{song2004estimation}. This is illustrated by a simulation study in Section~\ref{MixedCase}. We compare the finite sample performance of different bootstrap methods under the general setup of mixed interval censoring model. These comparisons illustrate the superior performance of the smoothed bootstrap procedure.

Our results also shed light on the behavior of bootstrap methods in similar non-standard convergence problems, such as the monotone regression estimator \citep{B70},  Rousseeuw's least median of squares estimator  \citep{R84}, and the estimator of the shorth \citep{A72,SW86}; see also \cite{GW01} for statistical problems in which the distribution ${\mathbb C}$  arises.

\section{Current status model}\label{CSModel}

\subsection{A sufficient condition for the consistency of the bootstrap}\label{SuffCond}

Under the current status model,  our data are $\mathbf{Z}_n = \{(T_i,\Delta_i)\}_{i=1}^n$, where $\Delta_i = \mathbf{1}_{X_i \le T_i}$.  Each $X_i$ can be interpreted as the unobserved time of onset of a disease and $T_i$ is the check-up time at which the $i$th patient is observed. 
We assume that $X_i \sim F$ and $T_i \sim G$ are independent and that $F$ and $G$ are continuously differentiable at $t_0 > 0$ ($t_0$ being a point in the interior of the support of $F$) with derivatives $f(t_0) > 0$ and $g(t_0) > 0$.

We want to approximate the distribution function $H_n$ of  $$\gamma_n=n^{1/3} \{\tilde F_n(t_0) - F(t_0)\}$$ by using bootstrap methods. In our model based bootstrap approach we choose an estimator, say $F_n$, of $F$ (which could be NPMLE $\tilde F_n$ or a smoothed version of it) and generate the bootstrapped response values as $\Delta_i^* \sim \mbox{Bernoulli}( F_n(T_i))$, fixing the values of $T_i$. This is the analogue of ``bootstrapping residuals'' in our setup. Let $\tilde F^*_{n}$ be the NPMLE of the bootstrap sample.

In the following we establish conditions on $F_n$ such that the bootstrap procedure is consistent, i.e., $$\gamma^*_n = n^{1/3} \{\tilde F_n^*(t_0) -  F_n(t_0)\}$$ converges weakly to $\kappa\mathbb{C}$, as defined in (\ref{eq:limF_n1}), given the data.

We first start by formalizing the notion of {\it consistency} of the bootstrap. Let $H_n^*$ be the conditional distribution function of $\gamma_n^*$, the bootstrap counterpart of $\gamma_n$, given the data. Let $d$ denote the Levy metric or any other metric metrizing weak convergence of distribution functions. We say that $H_{n}^*$ is {\it weakly consistent} if $d(H_n,H_n^*)\stackrel{}{\rightarrow} 0$ in probability. If the convergence holds with probability 1, then we say that the bootstrap is strongly consistent. If $H_{n}$ has a weak limit $H$, then consistency requires $H_{n}^*$ to converge weakly to $H$, in probability; and if $H$ is continuous, consistency requires 
$\sup_{x \in \mathbb{R}} |H_{n}^*(x) - H(x)| \stackrel{}{\rightarrow} 0 \mbox{ in probability as } n \rightarrow \infty.$ 

 Let $ F_n$ be a sequence of distribution functions that converge weakly to $F$ and suppose that
 \begin{equation}\label{eq:F_n-F}
\lim_{n\rightarrow\infty}\| F_n-F\|=0,
\end{equation} 
 almost surely, where for any bounded function $h:I \rightarrow \mathbb{R}$, $\|h\| = \sup_{x \in I} |h(x)|$. 
As shown in \cite{GW92},  the NPMLE obtained from the bootstrap sample $\tilde F^*_n$, defined as the maximizer of \eqref{npmle} over all distribution functions is a step function with possible jumps only at the predictor values $T_1, \ldots, T_n.$ We have the following result on the consistency of bootstrap methods in the current status model.   
\begin{theorem}\label{thm2}
If~\eqref{eq:F_n-F} and
\begin{equation}\label{condf}
\lim_{n\rightarrow\infty} n^{1/3}|F_n(t_0+n^{-1/3}t)-F_n(t_0) - f(t_0)n^{-1/3}t|=0
\end{equation} 
hold almost surely, then  conditional on the data, the bootstrap estimator $\gamma_n^*$ converges in distribution to $\kappa\mathbb{C}$, as defined in~\eqref{eq:limF_n1},  almost surely.
\end{theorem}


\subsection{Inconsistency of bootstrapping from $\tilde F_n$}\label{InconsB}
Consider the case when we bootstrap from the NPMLE $\tilde F_n$. Conditional on the predictor $T_i$, we generate the bootstrap response $\Delta_i^* \sim \mbox{Bernoulli}(\tilde F_n(T_i))$. Thus we take $F_n = \tilde F_n$ 
 and approximate the sampling distribution of $\gamma_n$ by the conditional distribution of 
$ \gamma_n^* = n^{1/3} \{\tilde F_n^*(t_0) - \tilde F_n(t_0)\},$ given the data. 
For this bootstrap procedure to be consistent, the conditional distribution of $\gamma_n^*$ must converge to that of $\kappa \mathbb{C}$, in probability. 
\begin{theorem}\label{thm:InconsBootsNPMLE}
    Unconditionally  $\gamma_n^*$ does not converge in distribution to  $\kappa \mathbb{C}$, and thus the bootstrap method is inconsistent. 
\end{theorem}

In fact it can be argued, as in~\cite{sen2010inconsistency}, that conditionally, $\gamma_n^*$ does not have any weak limit in probability. The inconsistency of bootstrapping from the NPMLE results from the lack of smoothness of $\tilde F_n$. At a more technical level, the lack of the smoothness manifests itself through the failure of equation~(\ref{eq:limII}) in  the Appendix. 

We illustrate through a simulation study the inconsistency of the NPMLE bootstrap method. 
The upper panel of Table \ref{table1} gives the estimated coverage probabilities of nominal 90\% confidence intervals for $F(1)$, where the true distribution of $F$ is assumed to be Exp(1), or the folded normal distribution,  $|N(0,1)|$, and $G$ is taken as the uniform distribution on $[0,2]$. We use 500 bootstrap samples to compute each confidence interval and construct 500 such  intervals.  
Throughout, we adopt this bootstrap setup unless otherwise specified. 
Table \ref{table1} shows that the coverage probabilities are much smaller than the nominal 90\% value and there is no  significant improvement as the sample size increases.

\begin{table}[h]
\begin{center}
\caption{\it Estimated coverage probabilities of nominal $90\%$ CIs for $F(1)$ for two distributions: Exp(1) and $|Z|$ with $Z\sim N(0,1).$ }{
\begin{tabular}{ccccc}
&$n$ & 100 & 200 & 500\\[3pt]
NPMLE &Exp(1) & 0.73  & 0.72   & 0.74 \\
&$ |N(0,1)|$ & 0.69 & 0.70  & 0.73  \\[3pt]
SMLE &Exp(1) & 0.89  &0.88  & 0.90 \\
&$ |N(0,1)|$ & 0.88  & 0.91  & 0.89  
\end{tabular}}
      \label{table1}
\end{center}
$\vspace{-0.3in}$
\end{table}

Furthermore, we compare the exact and bootstrapped distributions. Due to limitations of space, we only present results for $F$ being Exp(1). Figure~\ref{Figure1}(a) shows the distribution of $\gamma_n$, obtained from 10000 random samples of sample size 500, and its bootstrap estimate (that of $\gamma_n^*$) from a single sample based on 10000 bootstrap replicates. We see that the bootstrap distribution is different from that of $\gamma_n$. 
$\vspace{-0.3in}$
\begin{figure}[h!]

\center
\subfigure[NPMLE bootstrap]{\includegraphics[width=5.5cm]{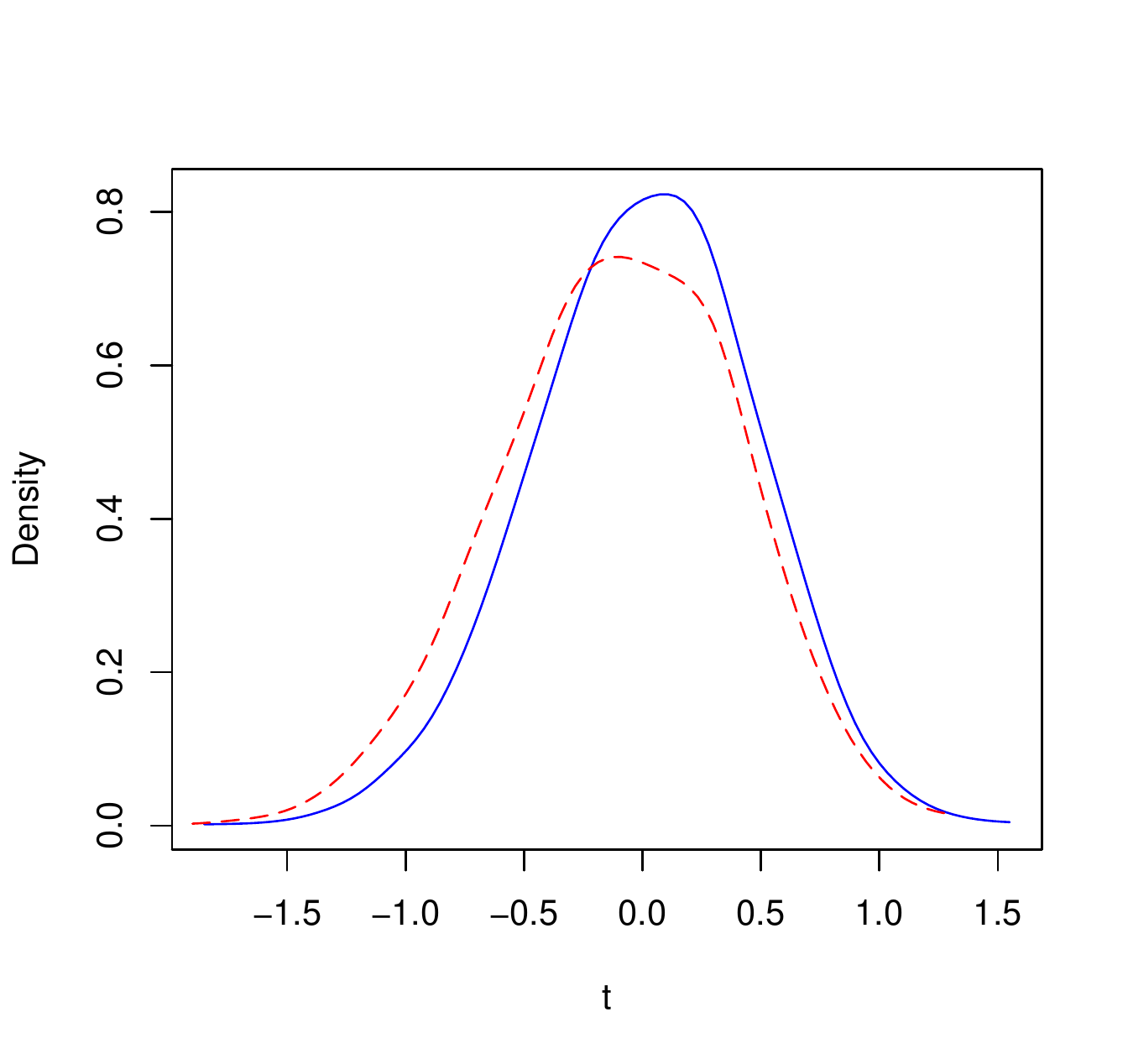}}
\quad
\subfigure[SMLE bootstrap]{\includegraphics[width=5.5cm]{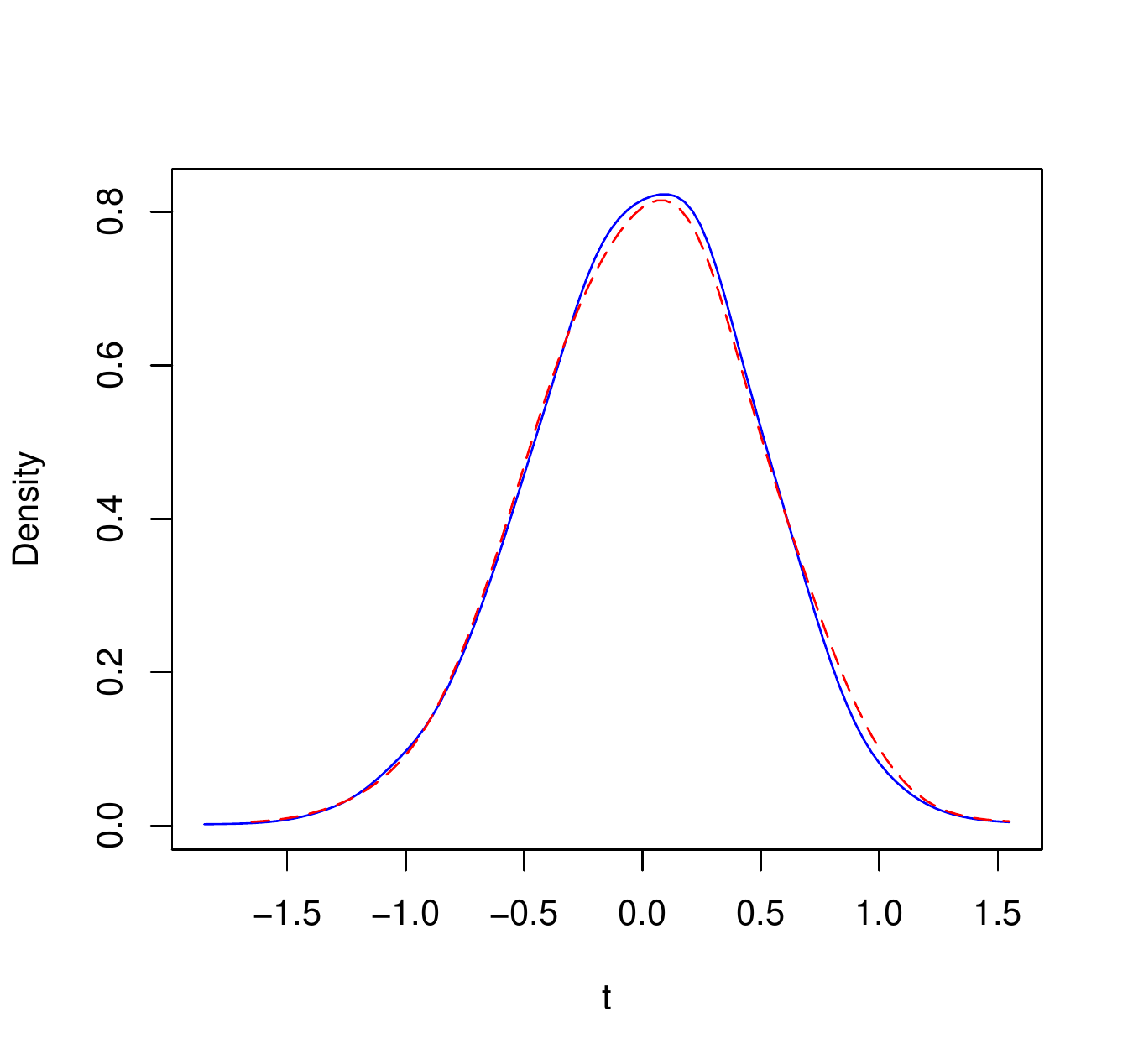}}
\caption{Estimated density functions of  $\gamma_n$ from 10000 Monte Carlo simulation (solid curve) and the bootstrap distribution of $\gamma^*_n$ when bootstrap samples are drawn from NPMLE $\tilde F_n$ (dashed, left panel) and SMLE $\check F_{n,h}$ with $h = 0.3$ (dashed, right panel).  $F$ is taken as Exp(1) and $n=500$.}
\label{Figure1}
\end{figure}

To illustrate the behavior of the conditional distribution of $\gamma_n^*$ we show in Figure \ref{Figure2}(a) the estimated $0.95$ quantiles of the bootstrap distributions for two independent data sequences as the sample size increases from 500 to 5000. The $0.95$ quantile of the limiting distribution of $\gamma_n$ is indicated by the solid line in each panel of Figure \ref{Figure2}.   
We can see that the bootstrap $0.95$ quantiles fluctuate enormously as the sample size increases from 500 to 5000 and do not converge to the $0.95$ quantile of $\kappa \mathbb{C}$. This gives strong empirical evidence that the bootstrapped $0.95$ quantiles do not converge.

\begin{figure}[h!]\center
\subfigure[NPMLE bootstrap]{\includegraphics[width=5.5cm]{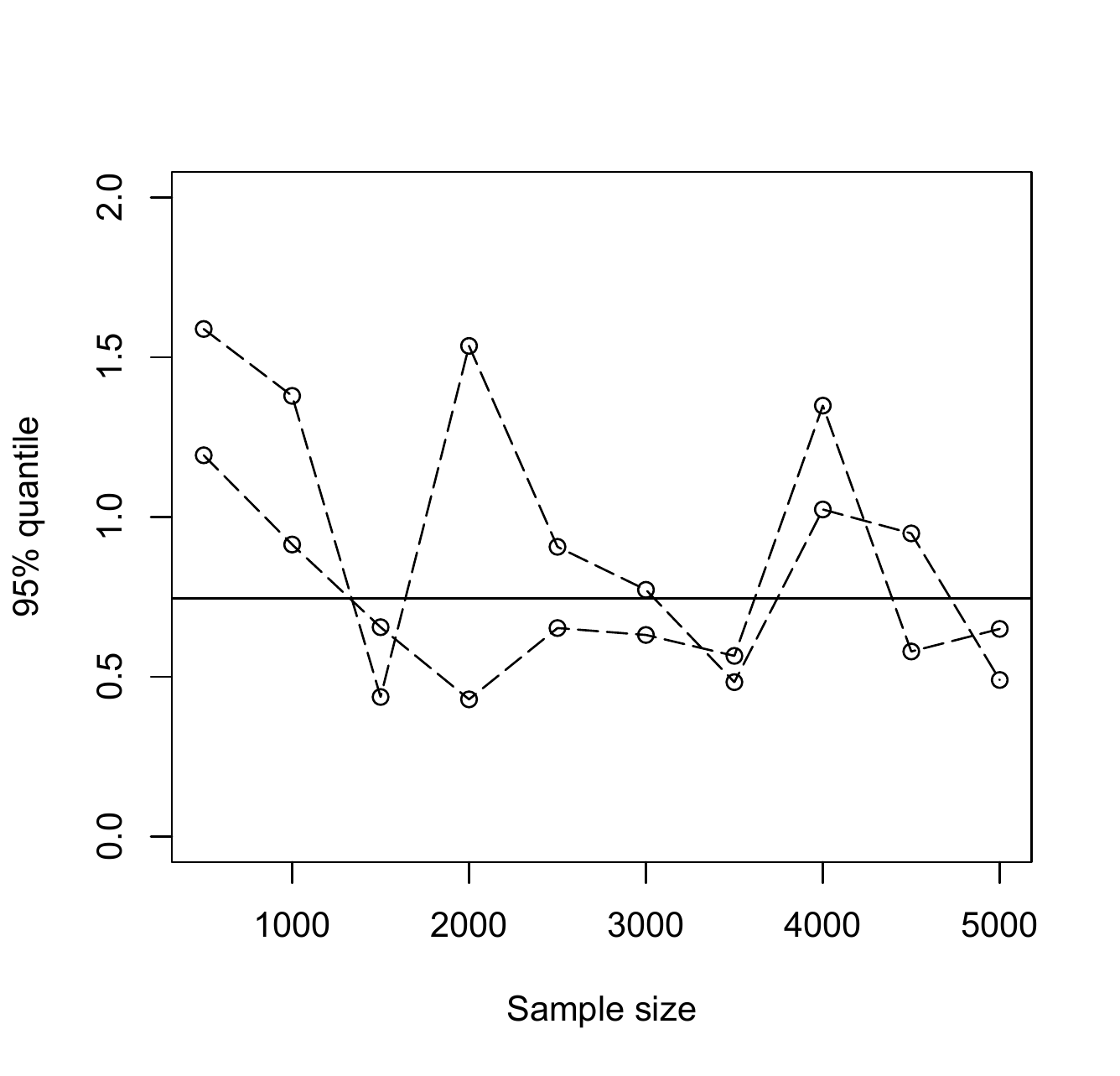}}
\quad
\subfigure[SMLE bootstrap]{\includegraphics[width=5.5cm]{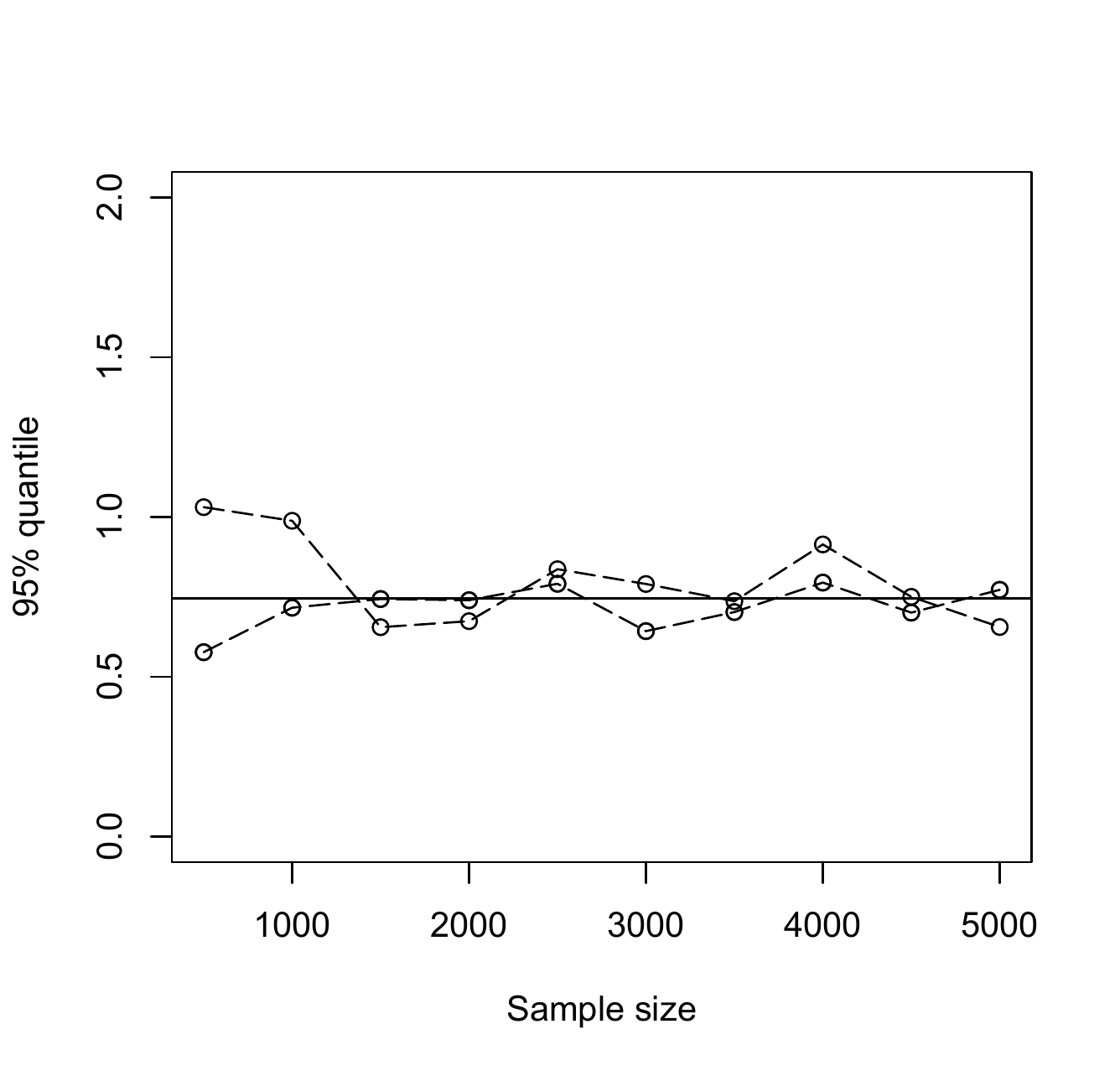}}
\caption{Estimated 0.95 quantiles of the bootstrap (dashed) and the limiting (solid) distribution.}
\label{Figure2}
\end{figure}

\subsection{Consistent bootstrap methods}\label{ConsSB}
We show that generating bootstrap samples from a suitably smoothed version of $\tilde{F}_n$ leads to a consistent bootstrap procedure. We propose the following smoothed estimator $\check F_n$ of $\tilde F_n$; see \cite*{groeneboom2010maximum}. Let $K$ be a differentiable symmetric kernel density with compact support (say $[-1,1]$) and let $\bar K(t) = \int_{-\infty}^t K(s) \, ds$ be the corresponding distribution function. Let $h$ be the smoothing parameter. Note that $h$ may depend on the sample size $n$ but, for notational convenience, in the following we write $h$ instead of $h_n$. Let  
$K_{h}(t) =  K(t/h)/h  \mbox{ and } \bar K_{h} (t)= \bar K(t/h).$
Then the smoothed maximum likelihood estimator (SMLE) of $F$ is defined as
\begin{equation}\label{eq:SmoothF_n}
\check F_{n}(t) \equiv \check F_{n,h}(t) = \int \bar K_{h}(t-s) \, d \tilde F_n(s).
\end{equation}
It can be easily seen that $\check F_{n,h}$ is a non-decreasing function, as for $t_2>t_1$, $\bar K_{h}(t_2-s)\geq \bar K_{h}(t_1-s)$ for all $s$.
 Throughout this paper, without further specification, we use the following kernel function to illustrate the performance of the SMLE bootstrap:
 \begin{eqnarray}\label{kernal}
K(t) \propto (1-t^2)^2{\mathbf 1}_{[-1, 1]}(t).
\end{eqnarray} 
$\check F_{n,h}$ is a smoothed version of the step function  $\tilde F_n$. As  discussed in the previous section, the lack of smoothness of  $\tilde F_n$ leads to  the inconsistency of the NPMLE bootstrap method. On the other hand, the SMLE successfully mimics the local behavior of $F$ at $t_0$, and consequently gives the desired consistency as shown in  Theorem~\ref{thm:ConsBoots} below. Recall that when bootstrapping from the SMLE $\check F_{n,h}$ our bootstrap sample is $\{(\Delta_i^*, T_i)\}_{i=1}^n$ where $\Delta_i^* \sim \mbox{Bernoulli}(\check F_{n,h}(T_i))$.

Following \cite{GW92}, we assume that the  point of interest $t_0$ is in the interior of the support of $F$, ${\cal S}=[0, M_0]$ with $M_0<\infty$, on which $F$ and $G$ have bounded densities $f$ and $g$ staying away from zero, respectively. Furthermore, density $g$ has a bounded derivative on ${\cal S}$.

\begin{theorem}\label{thm:ConsBoots}
Suppose $F$ and $G$ satisfy the conditions listed above. Given that $h\rightarrow 0$ and $n^{1/3}(\log n)^{-1}h\rightarrow\infty$,
 the conditional distribution of $n^{1/3} \{\tilde F_{n}^*(t_0) - \check F_{n,h}(t_0)\}$, given the data, converges  to that of $\kappa \mathbb{C}$, in probability. Thus, bootstrapping from $\check F_{n,h}$ is weakly consistent.
\end{theorem}
We use simulation to illustrate the consistency of the SMLE bootstrap procedure. The lower panel of Table \ref{table1} gives the estimated coverage probabilities of nominal 90\% confidence intervals for $F(1)$ (when $F$ is assumed to be Exp(1) or $|N(0,1)|$ and  $G$ is taken as the uniform distribution on $[0,2]$). 
Here we take bandwidth $h=0.3$. We see that the coverage probabilities are consistent with the nominal 90\% level. Figure \ref{Figure1}(b) compares the distributions of $\gamma_n$, obtained from 10000 random samples of size 500, and the SMLE bootstrap estimator from a single sample, when $F$ is Exp(1). In addition, Figure \ref{Figure2}(b) shows the estimated 0.95 quantiles of the bootstrap distributions for two independent data sequences. We see that for the SMLE bootstrap, the estimated 0.95 quantile is converging to the appropriate limiting value. This validates our theoretical result. 
\subsection{Choice of the tuning parameter in practice}\label{TunPara}
We propose a bootstrap-based method of choosing the smoothing bandwidth $h$, required for computing SMLE $\check F_{n,h}$. A commonly used criterion for judging the efficacy of bandwidth selection techniques is the mean squared error (MSE) for bandwidth $h$,
\begin{equation}\label{eq:IMSE}
\mbox{MSE}(h) =  E [ \{\check F_{n,h}(t_0) - F(t_0) \}^2].
\end{equation}
However, the above quantity is not directly computable since $F$ is unknown in applications. To overcome this difficulty, different procedures have been explored in the literature.  Among them, bootstrap is again one of the most widely used methods and we use it to estimate the MSE in our paper. This bootstrap approach works as follows. The idea is to approximate \eqref{eq:IMSE} by
\begin{equation}\label{eq:IMSE1}
\mbox{BMSE}(h) = \frac{1}{B} {\sum}_{i=1}^{B}  \{\check F^*_{n,h}(t_0) - F_n(t_0) \}^2,
\end{equation}
where $\check F^*_{n,h}(t_0)$ is constructed as in \eqref{eq:SmoothF_n} (with bandwidth $h$) from data $\{(T_i,\Delta^*_i)\}_{i=1}^n$ with $\Delta^*_i \sim \mbox{Bernoulli}(F_n(T_i))$,  for $i=1,\ldots,n$, and $B$ is a large number. Throughout, we take $B$ = 500. In the following we study two choices of $F_n$, the NPMLE and the SMLE, and show that the NPMLE does not give consistent estimates of MSE($h$) while the SMLE performs well. 

A natural choice of $F_n$ is $\tilde F_n$, the NPMLE based on the data. However, as shown in Figure~\ref{Figure3}(a), the estimated MSE curves for different data sets are not consistent with the true curve,  simulated from 500 independent samples. Here we use the kernel function  in \eqref{kernal}. For each MSE curve, we  approximate it using BMSE($h_i$) with $h_i =i/20$, for $ i=1,\ldots,20.$

\begin{figure}[h!]
\subfigure[MSE from NPMLE ]{\includegraphics[width=4.85cm]{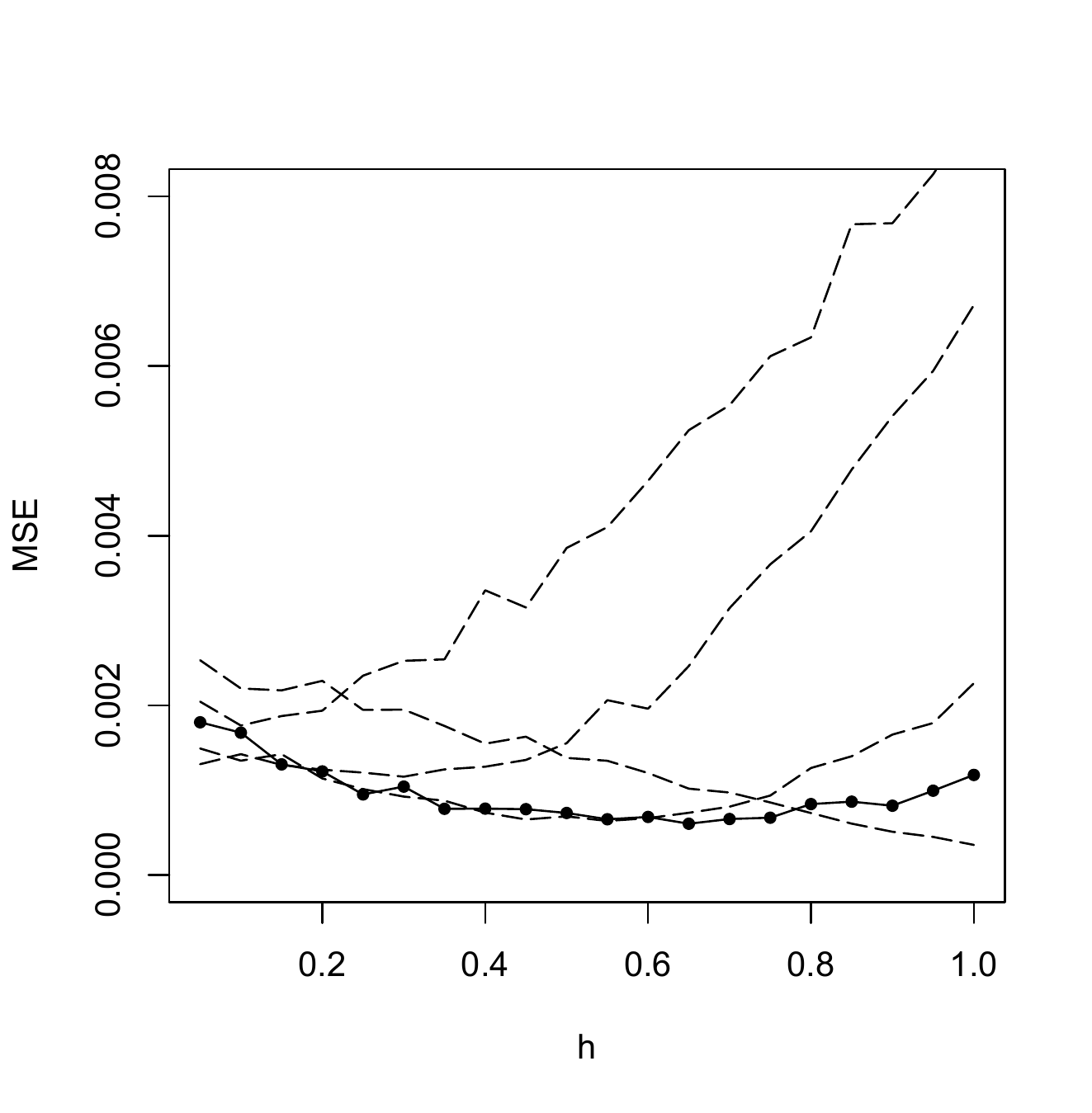}}
\subfigure[MSE from SMLE ]{\includegraphics[width=4.85cm]{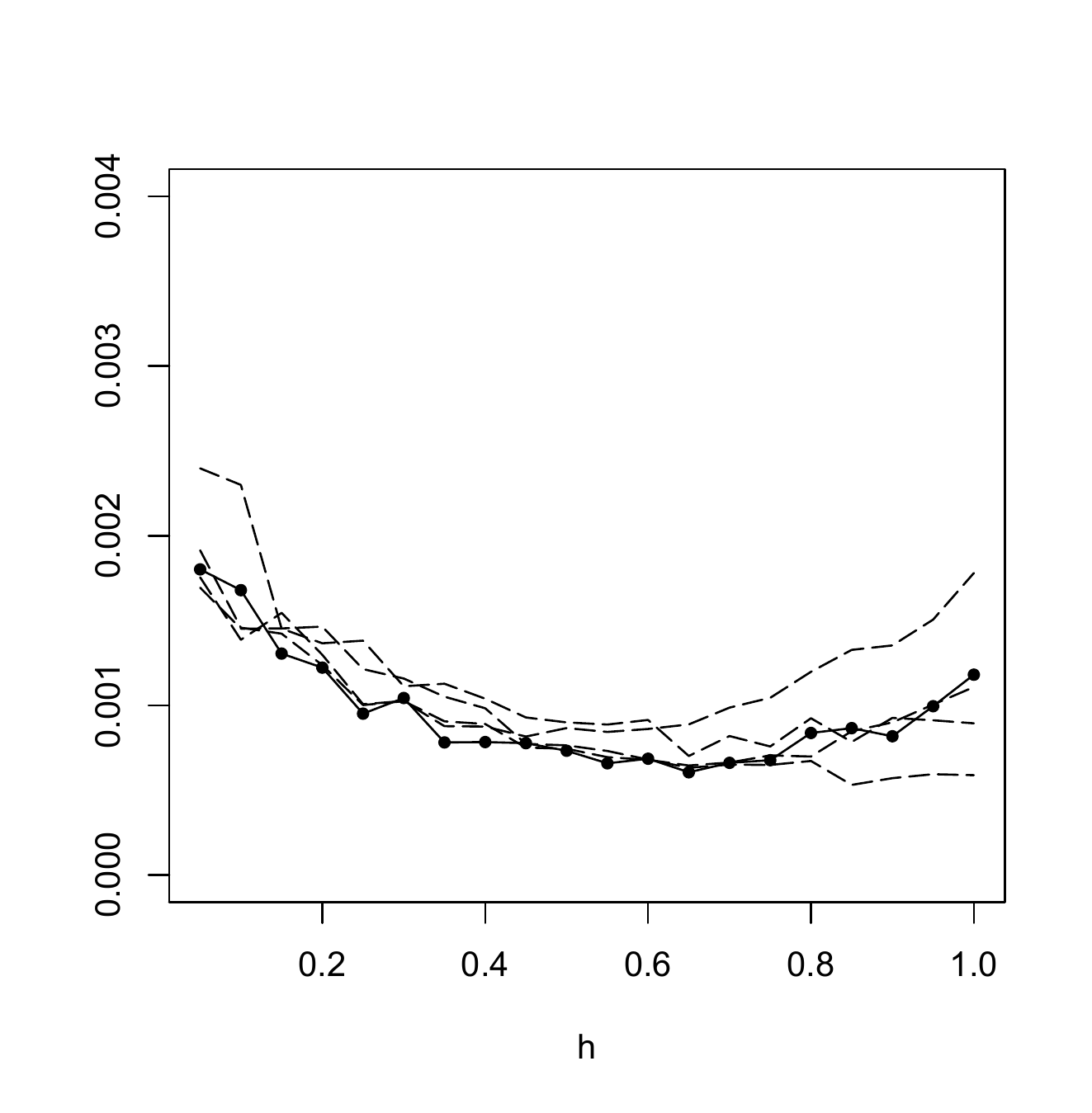}}
\subfigure[MSE with different $h_0$]{\includegraphics[width=4.85cm]{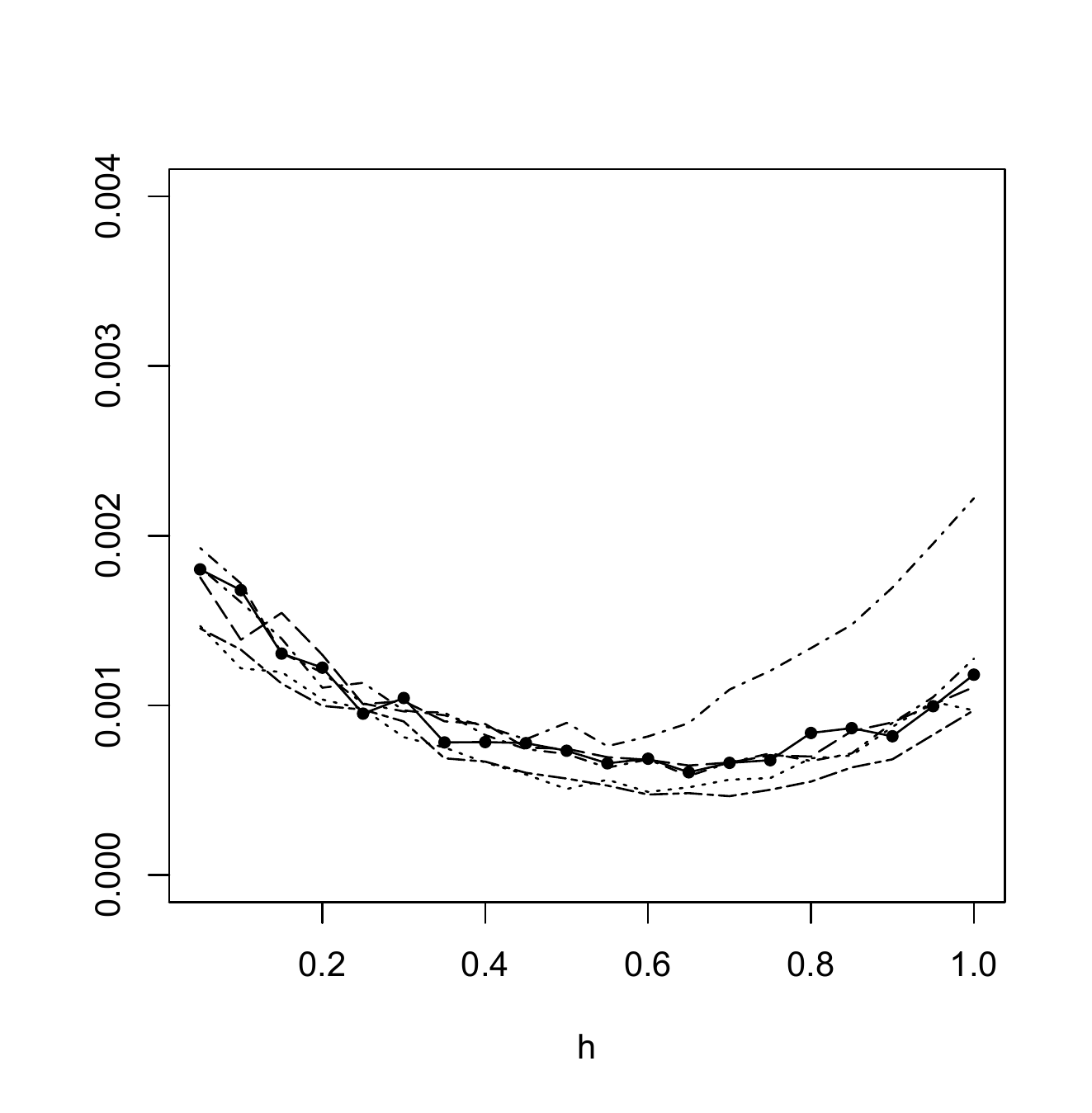}}
\caption{(a) Estimated MSE curves from the NPMLE bootstrap (dashed) and the true MSE based on 500 random samples (solid); 
(b) estimated MSE from the SMLE with pre-chosen bandwidth $h_0=0.5$ (dashed);
(c) estimated MSE  from the SMLE with $h_0=0.3,0.4,0.5,0.6$ and $0.7$ (dashed). 
 $F$ is taken as Exp(1) and $n=1000$.}
 $\vspace{-0.3in}$
\label{Figure3}
\end{figure}

%

Another choice of $F_n$ is $\check F_{n,h_0}$, the SMLE with a pre-chosen bandwidth $h_0$. This strategy is commonly used to select the bandwidth in density estimation; see, e.g., \cite{hazelton1996bandwidth} and \cite{gonzalez1996bootstrap}.
We choose $h_0$ as the initial bandwidth and sample $\Delta_i^* \sim \mbox{Bernoulli}(\check F_{n,h_0}(T_i))$. Then the MSE is estimated by
\begin{equation}\label{eq:IMSE1}
\frac{1}{B} {\sum}_{i=1}^{B} \{\check F^*_{n,h}(t_0) - \check F_{n,h_0}(t_0)\}^2.
\end{equation}
 Under the same simulation setup as for the NPMLE,  we show in Figure \ref{Figure3}(b) that the estimated MSE curves from the SMLE are  consistent with the true curve  based on 500 random samples.

A related issue in practice is how to choose the optimal initial smoothing bandwidth $h_0$. As in density estimation, where it has been shown that different initial values of bandwidths yield consistent estimation results, we also illustrate that different $h_0$ values give similar estimated MSE curves and therefore do not affect our final estimation much. We illustrate this through a simulation study.  We choose 5 initial values of $h_0$, $h_0=0.3, 0.4,0.5,0.6$ and $0.7$, and show in Figure \ref{Figure3}(c) that the estimated  MSE curves with  different $h_0$ values have similar shapes and are consistent with the true MSE curve. The minimum values of the estimated MSE curves are also close to the true minimum.

\section{Interval censoring, case 2}\label{case2}
In case 2 censoring, an individual is checked exactly at two time points $(T_1,T_2)$. Suppose that we have $n$ independent and identically distributed random vectors $\{(X_i,T_{i,1},T_{i,2})\}_{i=1}^n$, where for each pair $(X_i,T_{i,1}, T_{i,2})$, $X_i \sim F$ and $(T_{i,1}, T_{i,2})$ are independent and $T_{i,1}<T_{i,2}$. For the $i$th individual, we  observe $(T_{i,1}, T_{i,2}, \Delta_{i,1},\Delta_{i,2})$ where $\Delta_{i,1} = \mathbf{1}_{X_i \le T_{i,1}}  \mbox{ and }  \Delta_{i,2} = \mathbf{1}_{T_{i,1}<X_i \le T_{i,2}}$. Our goal is to estimate $F$ at time $t_0$, i.e., $F(t_0)$. 
The NPMLE $\tilde F_n$ for $F$ maximizes the log-likelihood function
\begin{eqnarray*}\label{npmle2}
{\sum}_{i=1}^n \left\{\Delta_{i,1} \log \mathbb{F}(T_{i,1}) +  \Delta_{i,2} \log(\mathbb{F}(T_{i,2}) - \mathbb{F}(T_{i,1})) +(1-\Delta_{i,1}-\Delta_{i,2})\log(1-\mathbb{F}(T_{i,2}))\right\}
 \end{eqnarray*}
over all distribution functions $\mathbb{F}$. Deriving the limiting distribution of $\tilde F_n$ in this case is quite difficult  and is still an open problem. \cite{groeneboom1991nonparametric} instead studied a one-step estimate $F^{(1)}_n$, obtained at the first step of the iterative convex minorant algorithm, starting the iterations from the underlying true distribution function $F$, and conjectured that  $F^{(1)}_n$  is asymptotically equivalent to the NPMLE. This conjecture is called {\it the working hypothesis} and is still unproved. Note that we cannot, of course, compute $F^{(1)}_n$ in practice. In the following, we assume that the working hypothesis holds and focus on bootstrapping the distribution of the one-step estimator. 

Let $H$ be the distribution function of observation times $(T_1, T_2)$ and assume that $F$ and $H$ are both differentiable at $t_0$ and $(t_0,t_0)$, respectively, with positive derivatives $f(t_0)$ and $h(t_0,t_0)$. From \cite{groeneboom1991nonparametric} and \cite{GW92}, we have that  
\begin{equation}\label{eq:limF_n2}
 (n\log n)^{1/3} \{F^{(1)}_n(t_0) - F(t_0)\} \stackrel{}{\rightarrow} \kappa_1 \mathbb{C},
\end{equation}
in distribution, where $\kappa_1 = \{\frac{3}{4} f(t_0)^2/h(t_0,t_0)\}^{1/3}$ and $\mathbb{C}$ is as defined in (\ref{eq:limF_n1}).

Under the working hypothesis,   the NPMLE also has the above limiting distribution in (\ref{eq:limF_n2}). 
Again, due to the nuisance parameters present in the limiting distribution, the above result cannot be directly applied to construct a confidence interval for $F(t_0)$. In the following we investigate the (in)-consistency of bootstrap methods and show that the smoothed model based bootstrap gives consistent result while bootstrapping from the NPMLE does not.

Let $ F_n$ be a sequence of distribution functions that converge weakly to $F$. We condition on $(T_{i,1},T_{i,2})$, and generate the bootstrap response $(\Delta_{i,1}^*, \Delta_{i,2}^*)$ by sampling $(\Delta_{i,1}^*, \Delta_{i,2}^*, 1-\Delta_{i,1}^*-\Delta_{i,2}^*) \sim \mbox{Multinomial}(1; F_n(T_{i,1}),  F_n(T_{i,2}) -F_n(T_{i,1}), 1-F_n(T_{i,2})).$
Then based on the bootstrap sample $\{(T_{i,1}, T_{i,2},\Delta^*_{i,1},\Delta^*_{i,2})\}_{i=1}^n$, we construct the one step NPMLE estimator $\tilde F^{*,(1)}_{n}(t)$, starting the iterations from $F_n$. 

\begin{theorem}\label{Thmcase2th}
For  a sequence of distribution functions $ F_n$ that converge weakly to $F$,  if  the following convergence holds, almost surely, uniformly on compacts (in $t$) \begin{equation}\label{condc2'}
\lim_{n\rightarrow\infty}(n\log n)^{1/3}|F_n(t_0+(n\log n)^{-1/3}t)-F_n(t_0) - f(t_0)(n\log n)^{-1/3}t|=0,
\end{equation} 
then,  the conditional distribution of $(n\log n)^{1/3} \{F_n^{*,(1)}(t_0)-F_n(t_0)\}$, given the data, converges  to $\kappa_1 \mathbb{C}$, almost surely.
\end{theorem}
The above theorem gives a sufficient condition for the bootstrap procedure to be consistent. In particular, for the SMLE bootstrap,  we have $F_n  = \check F_{n,h} $ and condition  (\ref{condc2'}) holds.  Let $\tilde F_{n,h}^{*,(1)}(t_0) $ be the corresponding one-step bootstrap estimator starting from  $\check F_{n,h}$ and let $\tilde F_n^*(t_0)$ be the NPMLE of the bootstrap sample. 
A similar argument as in the proof of Theorem \ref{thm:ConsBoots} gives 
that with properly chosen bandwidth $h$, conditionally,   $(n\log n)^{1/3}$ $ \{\tilde F_{n,h}^{*,(1)}(t_0) - \check F_{n,h}(t_0)\}$ converges in distribution to $\kappa_1 \mathbb{C}$, in probability. Then under the working hypothesis, conditionally, $(n\log n)^{1/3} \{\tilde F_{n}^{*}(t_0) - \check F_{n,h}(t_0)\}$ converges in distribution to $\kappa_1 \mathbb{C}$, in probability, and bootstrapping from $\check F_{n,h}$ is weakly consistent.
\begin{table}[h]
\begin{center}
\caption{\it Estimated coverage probabilities of nominal $90\%$ CIs for case 2 censoring.}{
\begin{tabular}{cccc}
$n$ & 100 & 200 & 500\\[3pt]
\hbox{NPMLE} &0.74   &0.76   &0.75\\
\hbox{SMLE} &0.88 & 0.89 & 0.91  
\end{tabular}}
\label{tablecase2}
\end{center}
 $\vspace{-0.3in}$
\end{table}

On the other hand, the NPMLE is a step function and does not satisfy \eqref{condc2'}, and therefore the above theorem is not applicable. In fact for bootstrapping from the NPMLE, the three convergence results in Lemma \ref{Thmcase2} in the Appendix may not hold. These theoretical arguments are supported by numerical results. Table \ref{tablecase2} shows that the NPMLE method has a low coverage rate for nominal $90\%$ confidence intervals while the smoothed bootstrap gives consistent results. Here  $F$ is taken as Exp(1) and $T_1, T_2$ are  the order statistics of two uniformly distributed random variables on $[0,2]$. For the SMLE bootstrap, the smoothing bandwidth $h$ is chosen as $n^{-1/5}$.
\section{Mixed case interval censoring}\label{MixedCase}
\subsection{NPMLE and SMLE bootstrap}
Under mixed case interval censoring, for an individual, we have $K$ observation times $T_{1}\leq \ldots\leq T_{K}$, where $K$ is an integer-valued random variable. 
 Let $\Delta_k=\mathbf{1}_{T_{k-1}<X\leq T_k}$, for $k = 1,\ldots,K$, with $T_0=0$. We observe $n$ independent and identically distributed copies of   $(T_k, \Delta_k: k=1,\ldots, K)$, i.e.,  $\{(T_{i,k}, \Delta_{i,k}: k=1,\ldots, K_i)\}_{i=1}^n$. Again we are interested in estimating $F$ at $t_0$. The NPMLE $\tilde F_n$ for $F$ maximizes the log-likelihood function
\begin{equation}\label{Mixnpmle}
 {\sum}_{i=1}^n\Big[{\sum}_{k=1}^{K_i} \Delta_{i,k} \log(\mathbb{F}(T_{i,k}) - \mathbb{F}(T_{i,k-1})) +\Big(1-{\sum}_{k=1}^{K_i} \Delta_{i,k}\Big)\log(1-\mathbb{F}(T_{i,K_i}))\Big]
 \end{equation}
 over all distribution functions $\mathbb{F}$. The limiting distribution of the NPMLE is unknown, and this complicates the problem of constructing confidence intervals for $F(t_0)$.  
 
In the following, we focus on the NPMLE $\tilde F_n$  and compare empirically the performance of different bootstrap methods in estimating the distribution of $n^{-1/3}(\tilde F_n(t_0)-F(t_0))$. We illustrate the inconsistency of bootstrapping from the NPMLE and the consistency of bootstrapping from a suitably smoothed NPMLE. As for each subject, only one $\Delta_{i,k}$, $k =1,...,K_i$, is 1, the computation of the NPMLE of mixed case interval censoring can be reduced to the case 2 interval censoring, as noted in \cite{huang1997interval} and \cite{song2004estimation}. This can be done efficiently by the iterative convex minorant algorithm; see \cite{GW92},  \cite{wellner1997hybrid}, and \cite{jongbloed1998iterative}. In this paper we use  R package ``Icens'' to estimate the NPMLE.

 \begin{table}[h]
\begin{center}
\caption{\it Estimated coverage probabilities of nominal  90\% CIs for mixed case interval censoring.}
\begin{tabular}{ccccc}
&$n$  & 100 & 200 & 500\\[3pt]
$\hbox{NPMLE}$& coverage &0.76  &0.77  & 0.76 \\
&length & 0.27 & 0.20  & 0.14 \\ [3pt]
$ \hbox{SMLE}$& coverage & 0.88 & 0.87  & 0.89  \\
&length & 0.25 & 0.20 & 0.14 \\ 
\end{tabular}
\label{table3}
\end{center}
\end{table}

Table \ref{table3} shows the estimated coverage probabilities of nominal $90\%$ confidence intervals when bootstrapping from the NPMLE and the SMLE. We can see that the performance of the SMLE is much better than that of the NPMLE in both coverage and  interval lengths. Here  $F$ is taken as Exp(1) and $K$ is chosen from the uniform distribution on set $\{1,2,3\}$. Given $K$, the observation times $(T_1,\ldots, T_K)$ are chosen as the $K$ order statistics from the uniform distribution on [0,2]. 
The smoothing bandwidth for the SMLE is chosen as $n^{-1/5}$.

\begin{figure}
\subfigure[MSE from NPMLE]{\includegraphics[width=4.8cm]{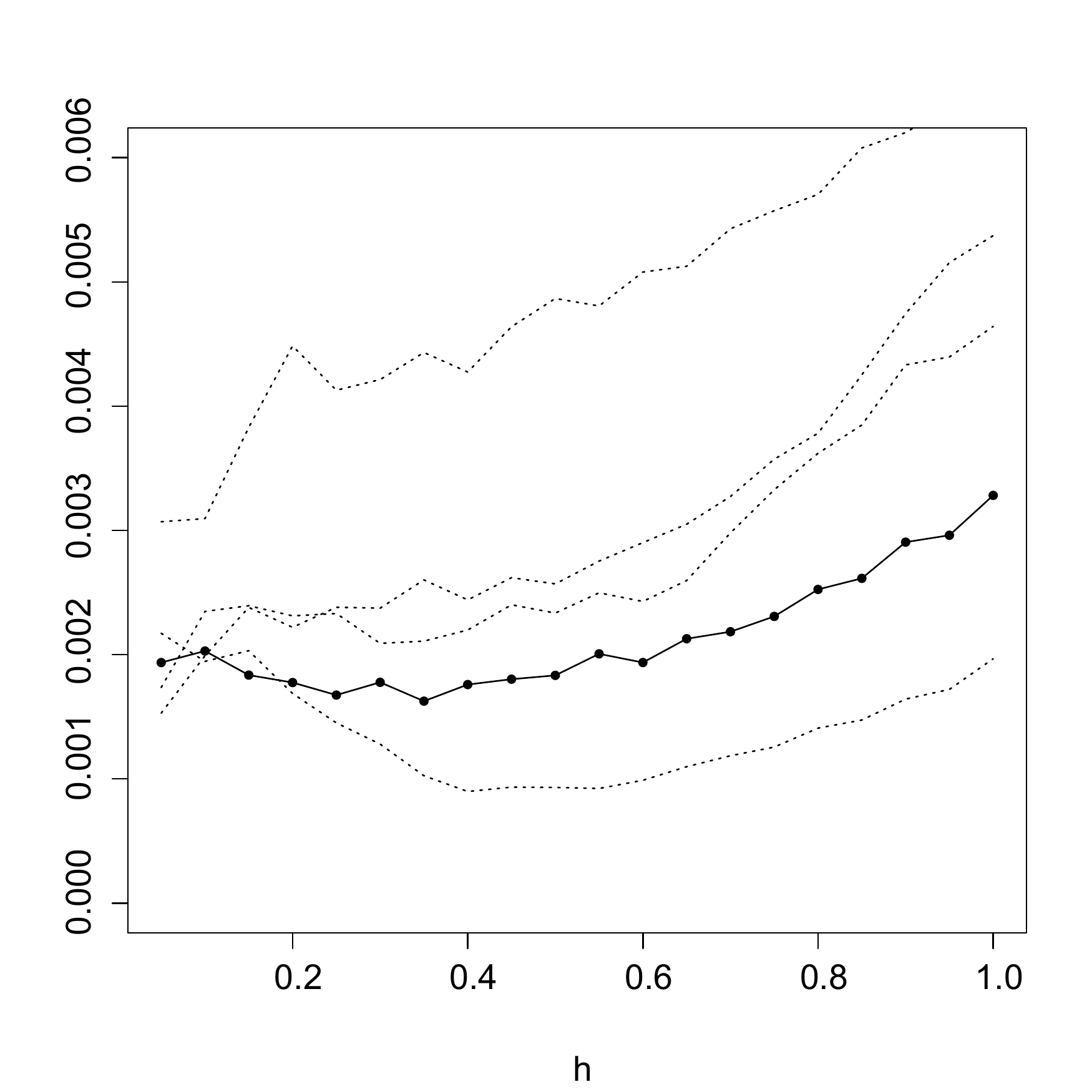}}
\subfigure[MSE from SMLE ]{\includegraphics[width=4.8cm]{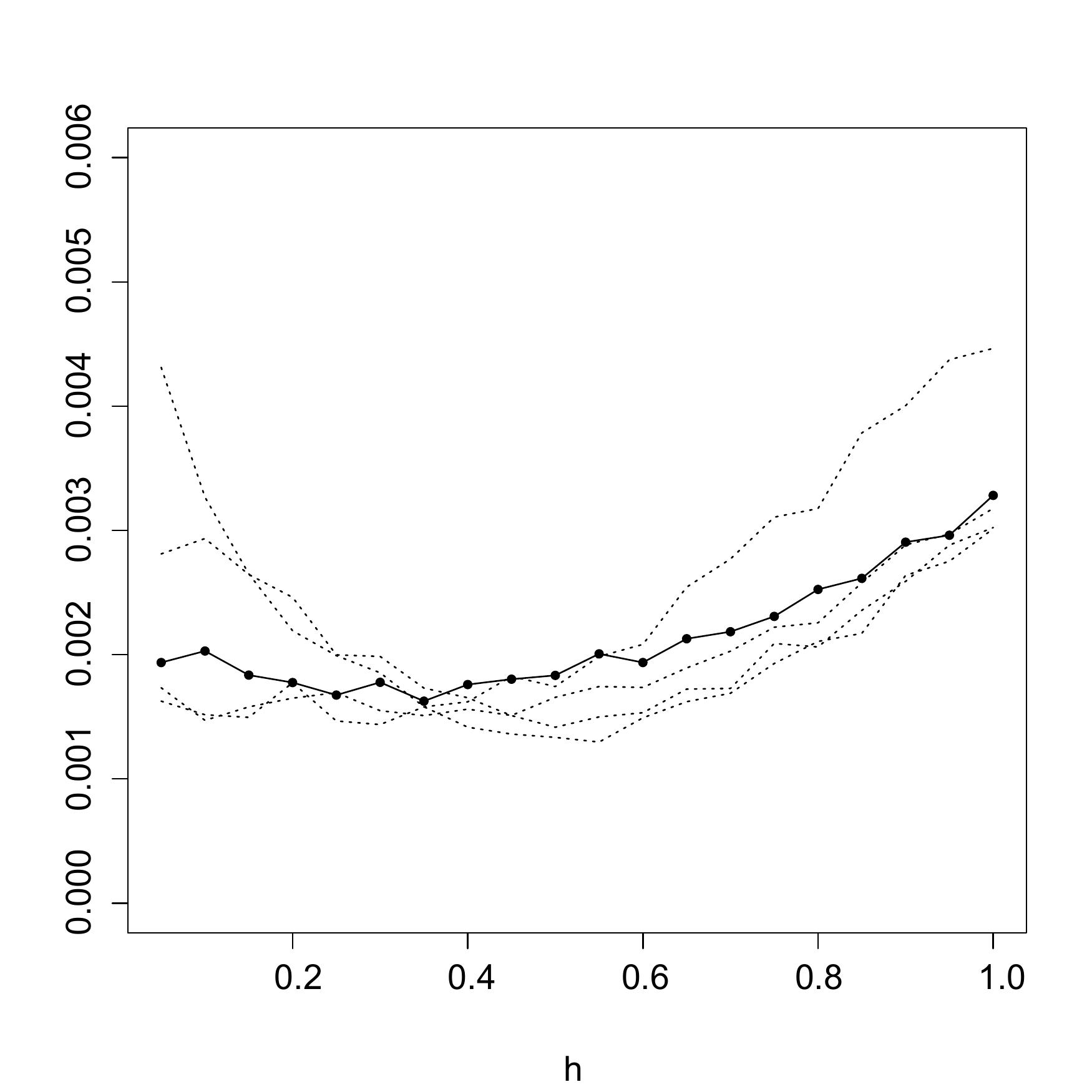}}
\subfigure[MSE with different $h_0$]{\includegraphics[width=4.8cm]{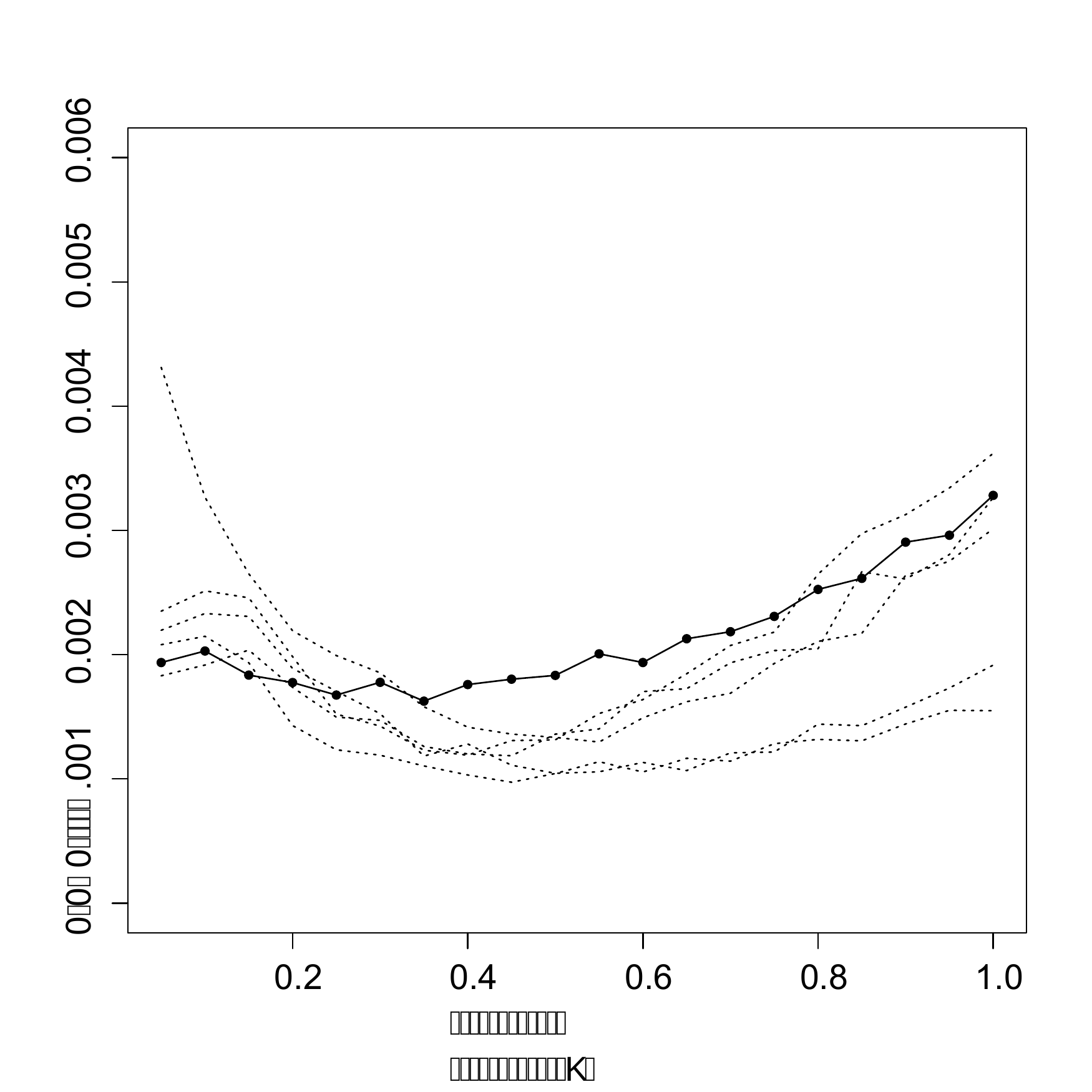}}
\caption{(a) The true MSE based on 500 random samples (solid) and estimated MSE from the NPMLE bootstrap (dashed); (b) MSE  from the SMLE with  $h_0=0.5$ (dashed); 
(c)  MSE from  the SMLE with $h_0=0.3,0.4,0.5,0.6$ and $0.7$ (dashed). }
\label{FigureMCS}
\end{figure}

We minimize the MSE criterion in (\ref{eq:IMSE}) to estimate the optimal bandwidth $h$ in the smoothed bootstrap method. Two choices of $F_n$, the NPMLE and a SMLE with an initial bandwidth $h_0$, are studied. 
The corresponding estimated MSE curves are displayed in Figure \ref{FigureMCS}. It can be clearly seen that curves based on the NPMLE are not consistent while those from the SMLE consistently estimate the true MSE curve. The effect of choosing different initial values of $h_0$ to obtain the optimal smoothing bandwidth is shown in Figure \ref{FigureMCS}(c). As in the current status model, we can see that the estimated curves are robust to the initial choice of $h_0$.

\subsection{Comparison with the existing methods}
To further illustrate the superiority of the proposed smoothed bootstrap method, we compare its finite sample performance with  the pseudolikelihood  method of~\cite{sen2007pseudolikelihood} and the $m$-out-of-$n$ bootstrap method of~\cite{LP06}. 

We present simulations from a mixed case censoring model under the same setup as in Section 3.1 of \cite{sen2007pseudolikelihood}. We take $F$ to be Exp(1) and $K$ is chosen uniformly from  $\{1,2,3,4\}$. Given $K$,  $(T_1,\ldots, T_K)$ are the $K$ order statistics from the uniform distribution on [0,3].  
We generate 1000 samples for each sample size shown in Table \ref{table_compare} and the corresponding $95\%$ confidence intervals for $F(\log 2)=0.5$ are constructed.  For the SMLE bootstrap, we choose bandwidths $h=0.5$ for $n=50, 100, 200, 500$ and $h=0.3$ for $n=1000, 1500, 2000$. Table \ref{table_compare} shows that the SMLE bootstrap gives more consistent results than the pseudolikelihood and $m$-out-of-$n$ methods in general.  The pseudolikelihood method is more anti-conservative while the $m$-out-of-$n$ method is in general conservative and has wider  intervals. 

\begin{table}[h!]
\begin{center}
\caption{\it Estimated coverage probabilities  of nominal $95\%$ CIs for mixed case interval censoring. The results for the pseudolikelihood (PL) and $m$-out-of-$n$ methods are taken from Table 1 in~\cite{sen2007pseudolikelihood}.}{
\begin{tabular}{llccccccc}
&$n$  & 50 &100 & 200 & 500&1000 &1500 &2000  \\[3pt]
$ \hbox{SMLE}$&coverage &0.92  & 0.94  &0.95  & 0.95 & 0.94  &0.95  & 0.95\\
&length &0.51  & 0.38   & 0.30 & 0.20  & 0.15  & 0.13  & 0.12\\[3pt]
$\hbox{PL}$ & coverage & 0.90 & 0.92 & 0.92  & 0.95 & 0.94 & 0.94 & 0.94 \\
&length & 0.41 & 0.33 & 0.26 & 0.20 &0.16 & 0.14 &0.12  \\[3pt]
$\hbox{$m$-out-of-$n$}$ & coverage & 0.97 & 0.97 & 0.96 & 0.96 & 0.95 & 0.96 & 0.97\\
& length & 0.54 & 0.47 & 0.31 & 0.24 & 0.17 & 0.16 & 0.14\\
\end{tabular}}
      \label{table_compare}
\end{center}
$\vspace{-0.6in}$
\end{table}

\bigskip
\section{Real data analysis}
\cite{finkelstein1985semiparametric} considered an interval censored data set of a study of early breast cancer patients. Between 1976 and 1980,   94 patients had been treated at the Joint Center for Radiation Therapy in Boston. They were assigned into two groups: one group treated with primary radiation therapy and adjuvant chemotherapy (48 patients) and the other group with radiotherapy alone (46 patients).  
Times of cosmetic deterioration,  determined by the appearance of breast retraction, were compared between the two treatment groups to determine whether chemotherapy has an impact on the rate of deterioration of the cosmetic state. 
The patients were checked at clinic every 4 to 6 month. For each patient, the only information available is a time interval when the retraction was present. See Section 5.3 in \cite{finkelstein1985semiparametric} and their Table 4 for more details, where all patients' time intervals are provided. 

We are interested in estimating the distribution of the retraction time ($F$). We model the data set as case 2 censoring and estimate the distribution functions separately for the two treatment groups, i.e., radiotherapy group ($T=0$) and radiotherapy and chemotherapy group ($T=1$). 
Figure \ref{datafit}(a) shows the NPMLE of $F$  computed for the two treatment groups. The distribution function of the group with $T=1$ dominates that of $T=0$ in general, which indicates that patients receiving radiotherapy and chemotherapy have an earlier deterioration time, as measured by the appearance of breast retraction. 
From Figure~\ref{datafit}(b) we see that the MSE curves from different $h_0$ are consistent with each other,  and their values decrease as $h$ increases from 1 to about 10 and then stay quite close. Based on these observations, we choose our bandwidth $h=10$. 

\begin{figure}[h]
\subfigure[NPMLE of $F$]{\includegraphics[width=5cm,height=5.15cm]{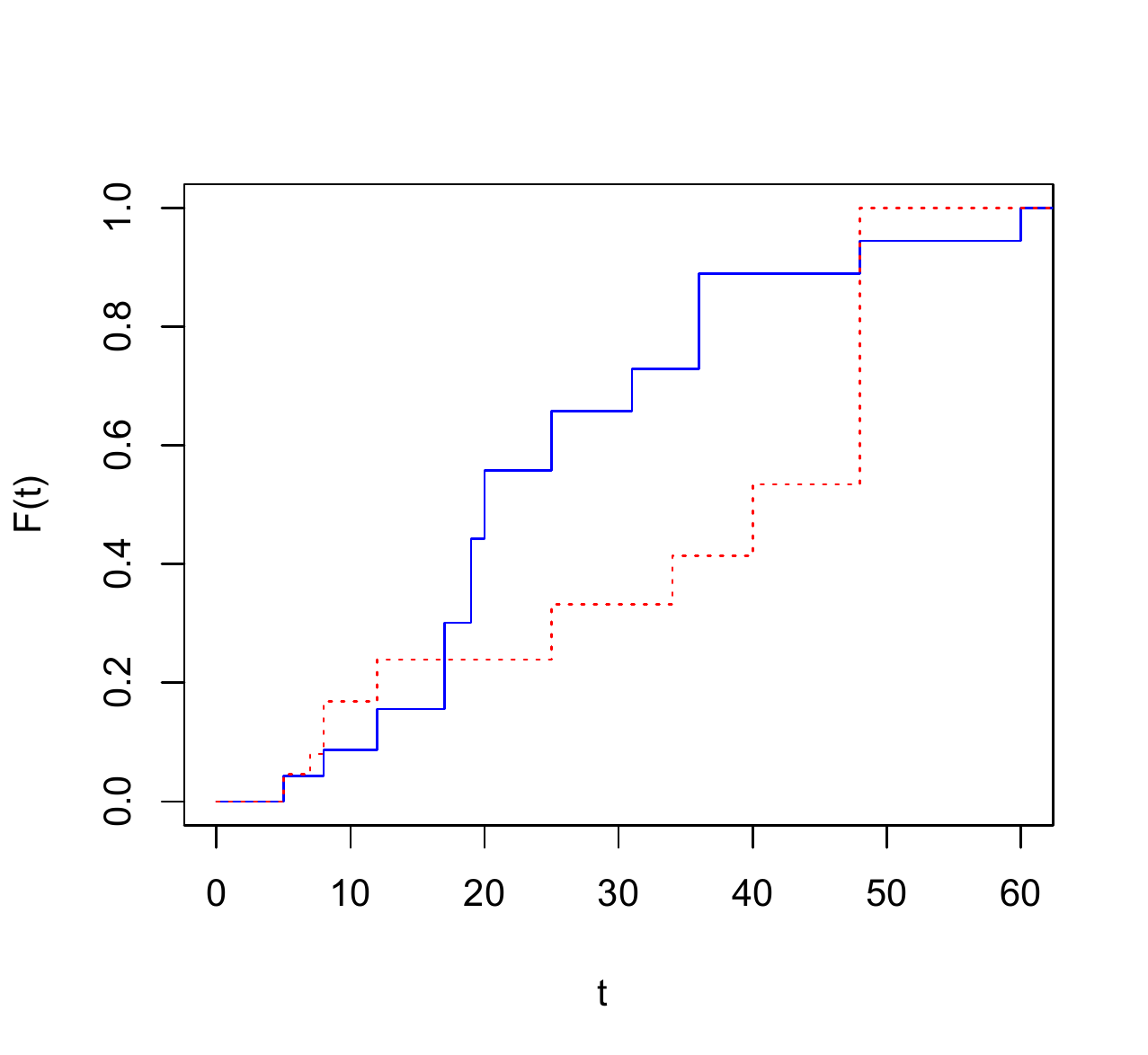}}
\subfigure[MSE curves with different initial bandwidths]{\includegraphics[width=10cm,height=5.2cm]{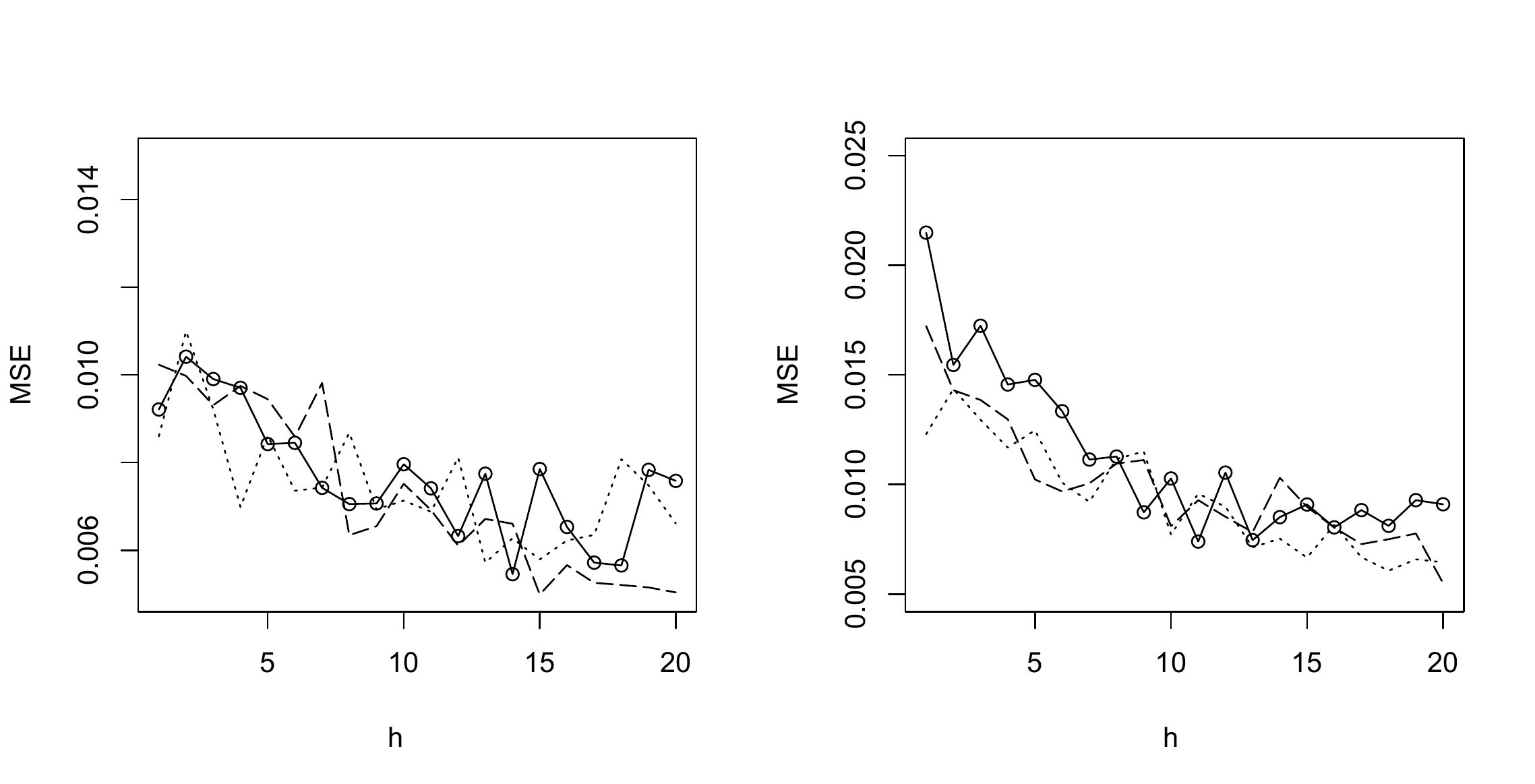}}
\caption{(a)  NPMLE of $F$ for groups $T=0$ (dashed) and $T=1$ (solid);
(b) estimated MSE curves at $t=30$  from the SMLE bootstrap with initial bandwidths $h_0=5, 10, 15$. The left plot is for $T=0$  and the right for $T=1$. }
\label{datafit}
\end{figure}
\begin{table}[h!]
\caption{\it CIs of  the distribution of  retraction time in  groups $T=0$ and $T=1$ at $t_0=20$ and 30.}{\small
\begin{tabular}{llcccccc}
 &&$\tilde F_n(20)$ & 90\% CI  & 95\% CI  &$\tilde F_n(30)$ & 90\% CI  & 95\% CI  \\[3pt]
{$T=1$}
&{\small SMLE} & 0.56 & [0.35, 0.79]  & [0.31, 0.82] & 0.66 & [0.48, 0.91] & [0.45, 0.94] \\
&{\small NPMLE }& 0.56 &  [0.41, 0.88]  &  [0.39, 0.92] & 0.66 & [0.44, 0.83] & [0.41, 0.86] \\[3pt]
{$T=0$}&{\small SMLE} & 0.24 & [0.11, 0.38]   & [0.08, 0.39] & 0.33  & [0.17, 0.50] & [0.14, 0.53] \\
&{\small NPMLE } & 0.24 &  [0.12, 0.35] & [0.09, 0.38] & 0.33 & [0.16, 0.49] & [0.13, 0.51]   
\end{tabular}}
\label{datacof}
\end{table}
Table~\ref{datacof} presents $90\%$ and $95\%$ confidence intervals of $F$ at $t_0=20$ and $30$ obtained using the NPMLE and SMLE bootstrap methods. The left extremities of the confidence intervals for the distribution function in group $T=1$  are  shifted to the right of those for the corresponding time points in group $T=0$, which indicates the treatment effect. We also see that the confidence intervals from SMLE and NPMLE bootstrap methods are quite different, which may be due to the inconsistency of the NPMLE method.

\section{Appendix: Proof of Theorems}
In this paper stochastic processes are regarded as random elements in $D(\mathbb{R})$, the space of right continuous functions on $\mathbb{R}$ with left limits, equipped with the projection $\sigma$-field and the topology of uniform convergence on compacta; see \cite{Pollard84}, Chapters IV and V for background.

For random elements $(V_n)_{n=1}^\infty$ and $V$ taking values in a metric space $(\mathfrak{X},\rho)$ we say that $V_n$ converges conditionally (given the data) in probability to $V$, almost surely (in probability), if for any given $\epsilon>0$, $ {P}(\rho(V_n,V)>\epsilon \mid {\mathbf Z}_n)\xrightarrow[]{} 0$ almost surely (in probability), where ${\mathbf Z}_n$ denotes our observed data.

\subsection{Proof of Theorem~\ref{thm2}}
We denote our bootstrap sample by $(T_1,\Delta^*_1),\ldots,(T_n,\Delta^*_n)$. 
Let $\mathbb{P}_n^*$ denote the induced measure of the bootstrap sample and write 
$$\mathbb{P}_n^* f(\Delta,T) = \frac{1}{n} \sum_{i=1}^n f(\Delta_i^*, T_i).$$
 Letting  $A = \{(x,t): x \le t\}$, we define the following stochastic processes:
\begin{eqnarray}
V_n^*(t) = \mathbb{P}_n^* \mathbf{1}_A \mathbf{1}_{\mathbb{R} \times [0,t]} 
= \frac{1}{n} \sum_{i=1}^n \Delta_i^* \mathbf{1}_{T_i \le t}, \quad
G_n^*(t) = \mathbb{P}_n^* \mathbf{1}_{\mathbb{R} \times [0,t]} 
= \frac{1}{n} \sum_{i=1}^n \mathbf{1}_{T_i \le t}. \nonumber
\end{eqnarray}
Let $\mathbb P_{T,n}$ be the empirical probability measure of  $\{T_i\}_{i=1}^n$.
Let $P_n$ be the probability measure induced by  $F_n$ and $\mathbb P_{T,n}$.
Note that under the conditional bootstrap procedure, $\mathbb P_n^* f(T) = P_n f(T)= \mathbb P_{T,n} f(T)=\sum_{i=1}^n f(T_i)/n.$ We use ${ E_n}$ to denote the expectation with respect to $ P_n$.

Appealing to the characterization of $\tilde F_n^*$ \citep[see pp.~298-299 of][]{VW00}, we know that 
\begin{equation}\label{vw00}
 \tilde F_n^*(t) \le a \,\,\mbox{  iff  } \,\, \arg \min_s \{V_n^*(s) - a G_n^*(s)\} \ge T_{(t)}
\end{equation}
where $T_{(t)}$ is the largest observation time that does not exceed $t$. By~\eqref{vw00}, the event that $n^{1/3} \{\tilde F_n^*(t_0) -  F_n(t_0)\} \le x$ is equivalent to $$ \arg \min_s \left\{ V_n^*(s) - [x n^{-1/3}+ F_n(t_0)] G_n^*(s) \right\} \ge T_{(t_0)}.$$ This is the same as $$ n^{1/3} \left[ \arg \min_s \left\{ V_n^*(s) - [x n^{-1/3} +  F_n(t_0)] G_n^*(s) \right\} - t_0 \right] \ge n^{1/3} (T_{(t_0)} - t_0).$$ Changing $s \mapsto t_0+tn^{-1/3}$ and using the fact that $n^{1/3} (T_{(t_0)} - t_0) = o(1)$, the above inequality can be re-expressed as  $$\arg \min_t \left[ V_n^*(t_0 + t n^{-1/3}) - \{x n^{-1/3} +  F_n(t_0)\} G_n^*(t_0 + t n^{-1/3}) \right] \ge o(1).$$
The left hand side of the above inequality can be written as
 \begin{eqnarray}
& & \arg \min_t \left[ \mathbb{P}_n^* \mathbf{1}_A \mathbf{1}_{\mathbb{R} \times [0,t _0+ t n^{-1/3}]} -  F_n(t_0) G_n^*(t_0 + t n^{-1/3}) - x n^{-1/3} G_n^*(t_0 + t n^{-1/3})   \right] \nonumber \\
& = & \arg \min_t \Big[ n^{2/3} \mathbb{P}_n^*\{\mathbf{1}_A -  F_n(t_0)\}(\mathbf{1}_{\mathbb{R} \times [0,t_0+tn^{-1/3}]} - \mathbf{1}_{\mathbb{R} \times [0,t_0]} )  \nonumber \\
& & \;\;\;\;\;\;\;\;\;\;\;\;\;\;\;\;\;\;  - x n^{1/3} [G_n^*(t_0 + t n^{-1/3}) - G_n^*(t_0)]   \Big]. \nonumber\\
&= & \arg \min_t \Big[ n^{2/3} \mathbb{P}_n^*(\mathbf{1}_A -  F_n(T))(\mathbf{1}_{\mathbb{R} \times [0,t_0+tn^{-1/3}]} - \mathbf{1}_{\mathbb{R} \times [0,t_0]} ) \nonumber \\
& &  \;\;\;\;\;\;\;\;\;\;\;\;\;\;\;\;\;\; +n^{2/3} \mathbb{P}_n^*( F_n(T) -  F_n(t_0))(\mathbf{1}_{\mathbb{R} \times [0,t_0+tn^{-1/3}]} - \mathbf{1}_{\mathbb{R} \times [0,t_0]} )\nonumber\\ 
 && \;\;\;\;\;\;\;\;\;\;\;\;\;\;\;\;\;\; -x n^{1/3} [G_n^*(t _0+ t n^{-1/3}) - G_n^*(t_0) ]\Big]  .\label{eqdec}
\end{eqnarray}
To study the distribution of $\gamma_n^*$, we start with the distributions  of the three terms in (\ref{eqdec}).
This is given in the following Lemma. 
\begin{lemma}\label{thm1}
We have the following convergence results:
\begin{itemize}
\item[\it{(i)}]We have that 
\begin{eqnarray}\label{eq:limIII}
    x n^{1/3} \{G_n^*(t_0 + t n^{-1/3}) - G_n^*(t_0)\} \rightarrow x g(t_0) t, 
        \end{eqnarray}
 uniformly on compacta, almost surely.

    \item[(ii)]  Let $\mathbb{Z}$ be a standard two-sided Brownian motion on $\mathbb{R}$ such that $\mathbb{Z}(0)= 0$. If \eqref{eq:F_n-F} holds, then, conditional on the data, the process
    \begin{eqnarray}\label{eq:limI}
    n^{2/3} \mathbb{P}_n^* \left\{ (\mathbf{1}_A -  F_n(T))(\mathbf{1}_{\mathbb{R} \times [0,t_0+tn^{-1/3}]} - \mathbf{1}_{\mathbb{R} \times [0,t_0]} ) \right\} \stackrel{d}{\rightarrow}
     \sqrt{F(t_0) [1 -F(t_0)] g(t_0)} \mathbb{Z}(t) \notag\\
    \end{eqnarray}
    almost surely in the space $D(\mathbb{R})$.

    \item[(iii)] If we have the following convergence uniformly on compacts in $t$
\begin{equation}\label{condf}
\lim_{n\rightarrow\infty} n^{1/3}|F_n(t_0+n^{-1/3}t)-F_n(t_0) - f(t_0)n^{-1/3}t|=0,
\end{equation} then conditionally 
    \begin{eqnarray}\label{eq:limII}
    n^{2/3} \mathbb{P}_n^* \left\{ ( F_n(T) -  F_n(t_0))(\mathbf{1}_{\mathbb{R} \times [0,t_0+tn^{-1/3}]} - \mathbf{1}_{\mathbb{R} \times [0,t_0]}) \right\} \stackrel{}\rightarrow \frac{1}{2} f(t_0) g(t_0) t^2
    \end{eqnarray}
uniformly on compacta, almost surely.
\end{itemize}
\end{lemma}

\begin{proof}[Proof of Lemma \ref{thm1}] 
(i)  To show the first convergence result, observe that
\begin{eqnarray}\label{eq:SimpIII}
& & x n^{1/3} \{G_n^*(t_0 + t n^{-1/3}) - G_n^*(t_0)\} \nonumber \\
& = & x n^{1/3} (\mathbb{P}_{T,n} - P)(\mathbf{1}_{\mathbb{R} \times [0,t_0 + t n^{-1/3}]} - \mathbf{1}_{\mathbb{R} \times [0,t_0]}) + x n^{1/3} P(\mathbf{1}_{\mathbb{R} \times [0,t_0 + t n^{-1/3}]} - \mathbf{1}_{\mathbb{R} \times [0,t_0]}). \nonumber 
\end{eqnarray}
By the law of iterated logarithm, this equals
$$ o(1) + x n^{1/3} \{G(t_0+ t n^{-1/3}) - G(t_0)\}  \rightarrow  x g(t_0) t, ~~a.s.,
$$
uniformly on compacta. \newline

\noindent(ii)
To show (\ref{eq:limI}), let $Z_{n,i} (t)= n^{-1/3} (\Delta_i^* - F_n(T_i)) W_{n,i}(t)$ where $W_{n,i}(t) = \mathbf{1}_{T_i \le t_0 + tn^{-1/3}} - \mathbf{1}_{T_i \le t_0}$.
The left-hand side of (\ref{eq:limI}) then can be expressed as $\sum_{i=1}^n Z_{n,i}(t)$. 
Note that $Z_{n,i}(t)$ has mean $0$ and variance 
$\sigma_{n,i}^2(t) = n^{-2/3} W_{n,i}^2(t) F_n(T_i) [1 - F_n(T_i)].$

 Therefore, for $h >0$, $s_n^2(t) := \sum_{i=1}^n \sigma_{n,i}^2(t)$ can be simplified as
$$ n^{1/3} \mathbb{P}_{T,n} [F_n(T) (1 - F_n(T)) (\mathbf{1}_{T\leq t_0 + t n^{-1/3}} 
- \mathbf{1}_{T\leq t_0})^2] .$$
The preceding display is equal to
\begin{eqnarray}
& & o(1) + n^{1/3} P [F_n(T) (1 - F_n(T))  (\mathbf{1}_{T\leq t_0 + tn^{-1/3}} 
- \mathbf{1}_{T\leq t_0})^2] \nonumber \\
& = & o(1) + n^{1/3} \int_{t_0}^{t_0 + t n^{-1/3}} F_n(u) (1 - F_n(u)) g(u) du \nonumber \\
& = & o(1) + \int_0^{t} F_n(t_0 + sn^{-1/3}) [1 - F_n(t_0 + sn^{-1/3})] g(t_0 + s n^{-1/3}) ds \nonumber \\
& \stackrel{a.s}{\rightarrow} & F(t_0) [1 - F(t_0)] g(t_0) t. \nonumber
\end{eqnarray}
By the Lindeberg-Feller CLT \cite[see pp. 359 of][]{B95} we have 
$$\sum_{i=1}^n Z_{n,i}(t) \stackrel{d}{\rightarrow} N(0,F(t_0) [1 - F(t_0)] g(t_0) t)$$
 for every $t \in \mathbb{R}$. Similarly we have the the convergence of the finite dimensional joint distribution. 
We only need to show the tightness of $\sum_{i=1}^n Z_{n,i}(t)$. 
By Theorem 15.6 in \cite{billingsleyconvergence}, it is sufficient to show that there exists a nondecreasing, continuous
function $H$ such that for any $t_1 < t < t_2$, $\gamma>0$ and $\alpha>1/2$
\begin{equation}\label{bll}
{ E_n} \Big[\Big|\sum_{i=1}^nZ_{n,i}(t_1)-\sum_{i=1}^nZ_{n,i}(t)\Big|^{\gamma}
\Big |\sum_{i=1}^nZ_{n,i}(t_2)-\sum_{i=1}^nZ_{n,i}(t)\Big| ^\gamma \Big] \leq  (H(t_2)-H(t_1))^{2\alpha}.
\end{equation}
Take $\gamma=2$ and $\alpha=1$.
Note that  the following inequality holds almost surely
\begin{eqnarray*}
&&{ E_n} \Big[\Big| \sum_{i}Z_{n,i}(t_1)-\sum_{i}Z_{n,i}(t)\Big|^2
\Big |\sum_{j}Z_{n,j}(t_2)-\sum_{j}Z_{n,j}(t)\Big|^2\Big] \\
&=&  n^{-4/3}
E_n\Big[\sum_{i} (\Delta_i^*-F_n(T_i))^2
\mathbf{1}_{t_0 + t_1n^{-1/3}<T_i \le t_0 + tn^{-1/3}} \\
&&\quad\quad\quad\quad \times\sum_{j} (\Delta_j^*-F_n(T_j))^2
 \mathbf{1}_{t _0+ tn^{-1/3}<T_j \le t_0 + t_2n^{-1/3}}\Big]\\
&\leq&  n^{-4/3}
\sum_{i}
\mathbf{1}_{t_0 + t_1n^{-1/3}<T_i \le t_0 + tn^{-1/3}} 
 \sum_{j}\mathbf{1}_{t _0+ tn^{-1/3}<T_j \le t_0 + t_2n^{-1/3}}\\
 &\rightarrow & g^2(t_0) (t_2-t_1)^2.
\end{eqnarray*}
Therefore we have the desired weak convergence result. 
\newline

\noindent(iii)
The left--hand side of (\ref{eq:limII}) can be  decomposed as
\begin{eqnarray}\label{eq:SimpII}
& & n^{2/3} (\mathbb{P}_n^* - P)\left[ (F_n(T) - F_n(t))(\mathbf{1}_{\mathbb{R} \times [0,t_0+tn^{-1/3}]} - \mathbf{1}_{\mathbb{R} \times [0,t_0]}) \right] \nonumber \\
& + &  n^{2/3} P \left[ (F_n(T) - F_n(t))(\mathbf{1}_{\mathbb{R} \times [0,t_0+tn^{-1/3}]} - \mathbf{1}_{\mathbb{R} \times [0,t_0]}) \right].
\end{eqnarray}
The first term in (\ref{eq:SimpII}) can be shown to converge to $0$ uniformly for $t \in [-K,K]$, in probability, for any $K> 0$, as follows:
\begin{eqnarray}
& & n^{2/3} \left|(\mathbb{P}_n^* - P)\left[ (F_n(T) - F_n(t_0))(\mathbf{1}_{\mathbb{R} \times [0,t_0+tn^{-1/3}]} - \mathbf{1}_{\mathbb{R} \times [0,t_0]}) \right] \right| \nonumber \\
& = & n^{2/3} \Big| \int_{t_0}^{t_0 + t n^{-1/3}} (F_n(u) - F_n(t_0))\; d(\mathbb{P}_{T,n} - P)(u) \Big| \nonumber \\
& = & n^{2/3} \Big| \Big[ (\mathbb{P}_{T,n} - P)(u) (F_n(u) - F_n(t_0))\Big]_{t_0}^{t_0 + t n^{-1/3}} -  \int_{t_0}^{t_0 +t n^{-1/3}}(\mathbb{P}_{T,n} - P)(u) \; dF_n(u) \Big| \nonumber \\
& \le & (\sqrt{n} \|\mathbb{P}_{T,n} - P \|) \; 2 n^{1/6}\left[ F_n(t_0 + K n^{-1/3}) - F_n(t_0-Kn^{-1/3}) \right] = o(1), ~~a.s.\nonumber
\end{eqnarray}
Therefore first term in (\ref{eq:SimpII}) is ignorable given 
that 
\eqref{condf} holds. The second term in (\ref{eq:SimpII}) can be simplified as:
\begin{eqnarray}
& &  n^{2/3} P \left[ (F_n(T) - F_n(t_0))(\mathbf{1}_{\mathbb{R} \times [0,t_0+tn^{-1/3}]} - \mathbf{1}_{\mathbb{R} \times [0,t_0]}) \right] \nonumber \\
& = & n^{1/3} \int_0^t [ F_n(t_0 + s n^{-1/3}) -  F(t_0)] g(t_0 + s n^{-1/3}) ds \nonumber \\
& = & (1+o(1)) n^{1/3} \int_0^t s n^{-1/3}  f(t_0) g(t_0 + s n^{-1/3}) ds \nonumber \\
& \stackrel{}{\rightarrow} & f(t_0) g(t_0) \frac{1}{2} t^2 ,\nonumber
\end{eqnarray}
where the last step follows from the assumption that $g$ is continuous at $t_0$.  This gives the desired conclusion.
\end{proof}

\medskip
We proceed to prove Theorem \ref{thm2}.\\

\begin{proof}[Proof of Theorem \ref{thm2}] This follows from a similar argument as in Section 3.2.15 in~\cite{VW00}.
We only need to show the uniform tightness of the minimum, i.e., for any $\epsilon$ and $B_0>0$,  there exists a constant $B$ such that 
$${ P}_n\left(\max_{x\in[-B_0,B_0]}n^{1/3} \left| \arg \min_s \left\{ V_n^*(s) - [x n^{-1/3} +  F_n(t_0)] G_n^*(s) \right\} - t_0 \right| >B\right)<\epsilon,~~a.s.$$
Here recall that $P_n$ is the probability measure induced by $F_n$ and  the empirical probability measure of  $\{T_i\}_{i=1}^n$.
The above tightness result follows from Theorem 3.4.1 in \cite{VW00}. In particular, following their notation, we take 
\begin{eqnarray*}
\mathbb{M}^*_n(h) &=&
 \mathbb{P}_n^*(\mathbf{1}_A -  F_n(T))(\mathbf{1}_{\mathbb{R} \times [0,t_0+h]} - \mathbf{1}_{\mathbb{R} \times [0,t_0]} ) \nonumber \\
& & + \; \mathbb{P}_n^*( F_n(T) -  F_n(t_0))(\mathbf{1}_{\mathbb{R} \times [0,t_0+h]} - \mathbf{1}_{\mathbb{R} \times [0,t_0]} )\nonumber\\ 
 &&- \; x n^{-1/3} [G_n^*(t_0 + h) - G_n^*(t_0) ],
 \end{eqnarray*}
 and $$\mathbb{M}_n(h)= P_n(\mathbf{1}_A -  F_n(T))(\mathbf{1}_{\mathbb{R} \times [0,t_0+h]} - \mathbf{1}_{\mathbb{R} \times [0,t_0]} ).$$
 Then the conditions of their Theorem 3.4.1 are satisfied with $\phi_n(\delta)=\sqrt{\delta}+x\delta n^{1/6}$. Together with the fact that $\mathbb{M}_n(h)$ converges to 0 in probability, which follows from  the results in Lemma \ref{thm1},  the tightness result holds and we have the weak convergence of $\gamma^*_n$. 
 
The above tightness result  and Lemma \ref{thm1} imply that conditional on the data $n^{1/3} \arg \min_s \{ V_n^*(s) - [x n^{-1/3} +  F_n(t_0)] G_n^*(s) \} - t_0 $ converges weakly to the process
$$T(x):=\hbox{argmin}_t\left\{  \sqrt{F(t_0) [1 -F(t_0)] g(t_0)} \mathbb{Z}(t) +\frac{1}{2} f(t_0) g(t_0) t^2- xg(t_0)t\right\}$$ almost surely.  
Therefore conditionally we have the following convergence  
\begin{eqnarray*}
P_n(n^{1/3} (\tilde F_n^{*}(t_0)-F_n(t_0))\leq x)
&\xrightarrow[]{a.s.}& P(T(x)\geq 0)\\
&=& P(T(x)-f(t_0)^{-1}x\geq -f(t_0)^{-1}x).
\end{eqnarray*}
By the stationary of process $T(x)-f(t_0)^{-1}x$ as shown in \cite{groeneboom1989brownian}, we have that $T(x)-f(t_0)^{-1}x$ and $T(0)$ have the same distribution function. Therefore, $P_n(n^{1/3} (\tilde F_n^{*}(t_0)-F_n(t_0))\leq x)$ converges almost surely to  
\begin{eqnarray*}
 && P(T(0)\geq -f(t_0)^{-1}x)\\
&= & P\Big(\hbox{argmin}_t\Big\{\sqrt{F(t_0) [1 -F(t_0)] g(t_0)} \mathbb{Z}(t) +\frac{1}{2} f(t_0) g(t_0) t^2 \Big\} \geq  -f(t_0)^{-1}x\Big)\\
&= & P\Big(\hbox{argmin}_t\Big\{ \mathbb{Z}\Big(\frac{f(t_0)^{2/3}g(t_0)^{1/3}}{[4F(t_0)(1-F(t_0)]^{1/3}}t\Big) +\Big(\frac{f(t_0)^{2/3}g(t_0)^{1/3}}{[4F(t_0)(1-F(t_0)]^{1/3}}t\Big)^2 \Big\} \geq -\kappa^{-1}x\Big).
\end{eqnarray*} 
Here recall that $\kappa = \{4 F(t_0)(1-F(t_0))f(t_0)/g(t_0)\}^{1/3}$ and the last equation is due to the Brownian scaling. 
Then, the above display is  equal to 
\begin{eqnarray*}
 P(\hbox{argmin}_t\{ \mathbb{Z}(t) +t^2 \} \geq -\kappa^{-1}x)
&=& P(\hbox{argmin}_t\{ \mathbb{Z}(t) +t^2 \} \leq \kappa^{-1}x),
\end{eqnarray*} 
which gives the desired conclusion.
\end{proof}

\subsection{Proofs of Theorems \ref{thm:InconsBootsNPMLE} and \ref{thm:ConsBoots}}
\begin{proof}[Proof of  Theorem \ref{thm:InconsBootsNPMLE}]
Proof of Theorem \ref{thm:InconsBootsNPMLE} follows a similar argument as in the proof of Theorem 3.1 in \cite{sen2010inconsistency} and we only show the key steps. The bootstrap NPMLE $\tilde F^*_n$ is the left derivative of the greatest convex minorant of the cumulative sum diagram consisting of  points $(G_n^*(t), V_n^*(t))$. 
Let $\mathbb{F}_n^*$ be the corresponding cumulative sum diagram function, i.e., for $u\in [G^*_n(T_{(i)}), G^*_n(T_{(i+1)}))$ function $\mathbb{F}^*_n(u)=V_n^*(T_{(i)})$, where 
 $T_{(1)}\leq T_{(2)} \leq \cdots\leq T_{(n)}$ are the order statistics of $T_1,\ldots,T_n$.
Then $\gamma^*_n$ equals the left derivative at $t=0$ of the greatest convex minorant of  process 
$$\mathbb{Z}^*_n(t):=n^{2/3}\{\mathbb{F}_n^*(G_n^*(t_0)+n^{-1/3}t))-\mathbb{F}_n^*(G_n^*(t_0))-\tilde F_n(t_0)n^{-1/3}t\}.$$
 We further write $\mathbb{Z}^*_n(t)$ as $\mathbb{Z}^*_{n,1}(t)+\mathbb{Z}^*_{n,2}(t)$, where
\begin{eqnarray*}
\mathbb{Z}^*_{n,1}(t) &:=&n^{2/3}\{(\mathbb{F}_n^*-\tilde{\mathbb{F}}_n)(G_n^*(t_0)+n^{-1/3}t))-(\mathbb{F}_n^*-\tilde {\mathbb{F}}_n)(G_n^*(t_0))\},\\
\mathbb{Z}^*_{n,2}(t)&:=&n^{2/3}\{\tilde{\mathbb{F}}_n(G_n^*(t_0)+n^{-1/3}t))-\tilde{\mathbb{F}}_n(G_n^*(t_0))-\tilde F_n(t_0)n^{-1/3}t\}.
\end{eqnarray*}
Here $\tilde{\mathbb{F}}_n$ is the greatest convex minorant of  the cumulative sum diagram function based on the observed data $(T_i, \Delta_i), i=1,\ldots, n$. These $\mathbb{Z}^*$ processes take analogous forms as $\mathbb{Z}_n^*(h)$, $\mathbb{Z}_{n,1}^*(h)$, $\mathbb{Z}_{n,2}^*(h)$ in Section 3.2 in \cite{sen2010inconsistency}. 
For $f$, let $L_{\mathbb R}f$ be its greatest convex minorant on $\mathbb R$. Following the proofs   of  Theorem 3.1 in  \cite{sen2010inconsistency} and Lemma \ref{thm1},
we have unconditionally 
\begin{eqnarray*}
\mathbb{Z}^*_{n,1}(t)  &\stackrel{d}{\rightarrow}&
 \mathbb U_1(t):=    \sqrt{F(t_0) [1 -F(t_0)] } \mathbb{Z}_1(t),\\
     \mathbb{Z}^*_{n,2}(t) &\stackrel{d}{\rightarrow}&
  \mathbb U_2(t):=   L_{\mathbb R}\mathbb{Z}^0_2(t)-L_{\mathbb R}\mathbb{Z}^0_2(0)
     -(L_{\mathbb R}\mathbb{Z}^0_2)'(0)t, 
\end{eqnarray*}
where 
$$\mathbb{Z}^0_2(t) = \sqrt{F(t_0) [1 -F(t_0)] } \mathbb{Z}_2(t)+ \frac{1}{2} f(t_0) g^{-1}(t_0) t^2,$$ 
and $ \mathbb{Z}_1(t)$ and  $ \mathbb{Z}_2(t)$  are two independent two-sided standard Brownian motions. Furthermore, we have that  the unconditional distribution of $\gamma^*_n$ converges to that of $ L_{\mathbb R} (\mathbb U_1 +\mathbb U_2)'(0)$, which is different from that of $\kappa\mathbb C$. This gives the inconsistency result. 
\end{proof}\\

\begin{proof}[Proof of Theorem \ref{thm:ConsBoots}]
We will apply Theorem~\ref{thm2}.  By Lemma 5.9 in \cite{GW92}, the NPMLE $\tilde F_n$ satisfies $\|\tilde F_n - F \| = O_p(n^{-1/3}\log n)$ under our assumption of $F$ and $G$. 
Since $h\rightarrow 0$ and $n^{1/3}(\log n)^{-1} ~h\rightarrow\infty$, we have the following holds uniformly in $t$:
 \begin{eqnarray*}
|\check F_{n,h}(t)-F(t)|&\leq & \Big|\int  \bar K_{h}(t-s)  d (\tilde F_n(s)-F(s)) \Big|+\Big| \int  \bar K_{h}(t-s) dF(s)-F(t)\Big|\\
&=& \Big| \int (\tilde F_n(s)-F(s)) d K_h(t-s) \Big|+ \Big| \int  \bar K_{h}(t-s) dF(s)-F(t)\Big|+o_p(1)\\ 
&= & O(1)\|\tilde F_n - F \| h^{-1} + o_p(1) =o_p(1).
\end{eqnarray*}
Thus \eqref{eq:F_n-F} holds in probability.

Next we show that (\ref{condf}) holds in probability for SMLE $\check F_{n,h}$, i.e., 
$$\lim_{n\rightarrow\infty} n^{1/3}|\check F_{n,h}(t_0+n^{-1/3}t)-\check F_{n,h}(t_0) - f(t_0)n^{-1/3}t|=0.$$
Note that we have 
\begin{eqnarray*}
 n^{1/3}(\check F_{n,h}(t_0+n^{-1/3}t)-\check F_{n,h}(t_0)) 
&= & n^{1/3} \int \left \{ \bar K_{h}(t_0+n^{-1/3}t-s) - \bar K_{h}(t_0-s)  \right\} d \tilde F_n(s) \\  
&= & o_p(1)+  t\int  K_{h}(t_0-s) \, d \tilde F_n(s). 
\end{eqnarray*}
Integrating by parts yields
\begin{eqnarray*}
\int  K_{h}(t_0-s) \, d \tilde F_n(s)& =& 
\int  K_{h}(t_0-s) \, d (\tilde F_n(s)-F(s))+\int K_{h}(t_0-s) dF(s)\\
&=& - \int (\tilde F_n(s)-F(s))dK_{h}(t_0-s) +  \int K_{h}(t_0-s) dF(s)+o_p(1)\\
&=& f(t_0) +o_p(1),
\end{eqnarray*}
given that $\|\tilde F_n - F\| = O_p(n^{-1/3}\log n)$, $h\to 0$ and $n^{1/3}(\log n)^{-1} h\rightarrow\infty$.
\end{proof}

\subsection{Proof of Theorem \ref{Thmcase2th}}
Let $\mathbb{P}_n^*$ denote the induced measure of the bootstrap sample.
For any $t>0$, define \begin{eqnarray*}
W^*_{n}(t) &=& W_{n,1}^*(t) +\int_0^t \{ F_n(s)-F_n(t_0) \}\, d W_{n,2}^*(s),
\end{eqnarray*}
where for $k=1$ and $2$,
\begin{eqnarray*}
W^*_{n,k}(t) &=& \int_{t_1\in[0,t],x\leq t_1} \frac{d{\mathbb P}^*_n(x,t_1,t_2)}{\{F_n(t_1)\}^k}
+\int_{t_1\in[0,t], t_1<x\leq t_2} \frac{d {\mathbb P}_n^*(x,t_1,t_2)}{\{F_n(t_1)-F_n(t_2)\}^k}\\
&&+\int_{t_2\in[0,t], t_1<x\leq t_2} \frac{d {\mathbb P}_n^*(x,t_1,t_2)}{\{F_n(t_2)-F_n(t_1)\}^k}
+\int_{t_2\in[0,t], x> t_2} \frac{d {\mathbb P}_n^*(x,t_1,t_2)}{\{F_n(t_2)-1\}^k}.
\end{eqnarray*}

Thanks to the characterization of $\tilde F_n^{*,(1)}$ \citep[see, e.g.,][]{groeneboom1991nonparametric}, we know that
$$P_n\left[(n\log n)^{1/3} \{\tilde F_n^{*,(1)}(t_0)-F_n(t_0)\}> x\right] = P_n [T^{*,(0)}_n(F_n(t_0)+(n\log n)^{-1/3}x)<t_0],$$
where $P_n$ is the  probability measure induced by $F_n$ and the empirical probability measure of $\{T_{i,1},T_{i,2}\}_{i=1}^n$, and 
$$T^{*,(0)}_n(x):=\hbox{sargmin}_t [W^*_{n}(t)- \{x-F_n(t_0)\}W_{n,2}^*(t)].$$
For a function $w(t)$, $\hbox{sargmin}_t w(t)$ means the maximum value of the minimizers of function $w(t)$. If there is a unique minimizer, then $\hbox{sargmin}_t w(t)=\hbox{argmin}_t w(t)$.
By the definition of $T^{*,(0)}_n(x)$, we can write \begin{eqnarray*}
&&(n\log n)^{1/3}\Big(T^{*,(0)}_n(F_n(t_0)+(n\log n)^{1/3}x)-t_0\Big)\\
&=& \hbox{sargmin}_t\Big\{n^{2/3}(\log n)^{-1/3}(W^*_{n}(t_0+(n\log n)^{1/3}t)-W^*_{n}(t_0))\\
&&- \; x (n\log n)^{1/3}(W_{n,2}^*(t_0+ (n\log n)^{-1/3}t)-W_{n,2}^*(t_0))\Big\}\\
&=& \hbox{sargmin}_t\Big\{n^{2/3}(\log n)^{-1/3}(W^*_{n,1}(t_0+(n\log n)^{1/3}t)-W^*_{n,1}(t_0))\\
&& + \; n^{2/3}(\log n)^{-1/3}\int_{t_0}^{t_0+(n\log n)^{1/3}t} (F_n(s)-F_n(t_0)) d W_{n,2}^*(s)\\
&&- \; x n^{1/3}(\log n)^{-2/3} (W_{n,2}^*(t_0+ (n\log n)^{-1/3}t)-W_{n,2}^*(t_0))\Big\}.
\end{eqnarray*}

As in the proof of Theorem \ref{thm1}, we start with the distributions of three terms in the above display. This is given by the following lemma. 
\begin{lemma}\label{Thmcase2}
For  a sequence of distribution functions $ F_n$ that converge weakly to $F$,  if  the following convergence holds uniformly on compacts (in $t$) \begin{equation}\label{condc2}
\lim_{n\rightarrow\infty}(n\log n)^{1/3}|F_n(t_0+(n\log n)^{-1/3}t)-F_n(t_0) - f(t_0)(n\log n)^{-1/3}t|=0,
\end{equation} then, conditionally on the data, we have the the following convergence almost surely uniformly on compacta
\begin{eqnarray*} 
n^{1/3}(\log n)^{-2/3}(W_{n,2}^*(t_0+ (n\log n)^{-1/3}t)-W_{n,2}^*(t_0))
&\xrightarrow[]{p}&  \frac{2}{3f(t_0)}h(t_0,t_0)t,\\
n^{2/3}(\log n)^{-1/3}\int_{t_0}^{t_0+(n\log n)^{1/3}t} (F_n(s)-F_n(t_0)) d W_{n,2}^*(s)
&\xrightarrow[]{p}& \frac{1}{3}h(t_0,t_0)t^2,\\
n^{2/3}(\log n)^{-1/3}(W^*_{n,1}(t_0+(n\log n)^{1/3}t)-W^*_{n,1}(t_0))
&\xrightarrow[]{d}& \sqrt{\frac{2}{3}h(t_0,t_0)/f(t_0)} \mathbb{Z}(t).
\end{eqnarray*}
\end{lemma}

\begin{proof}[Proof of Lemma \ref{Thmcase2}] 
Proof of Lemma  \ref{Thmcase2} follows from a similar argument as in  the proof of Theorem 5.3 in \cite{groeneboom1991nonparametric}.
Note that  for  $B_0>0$, if (\ref{condc2}) holds, we have that 
\begin{eqnarray*}
&&P_n\Big(T_1<X<T_2, ~T_1, T_2\in [t_0, t_0+B_0(n\log n)^{-1/3}]\Big)\\
&=&(1+o(1))\int_{t_0\leq t_1<t_2\leq t_0+B_0(n\log n)^{-1/3}}f(t_0)(t_2-t_1)h(t_0,t_0)dt_2dt_1\\
&=&(1+o(1))\frac{1}{6}f(t_0)h(t_0,t_0)B_0^3(n\log n)^{-1}, ~~a.s.,
\end{eqnarray*}
where $h$ is the density function of observation times $T_1$ and $T_2$. This mean that the event $T_1<X<T_2$ and both $T_1$ and $T_2$ are in interval $ [t_0, t_0+B_0(n\log n)^{-1/3}]$ has probability going to 0 at rate $n\log n$. Then, for $n$ observations $\{T_{i,1}, T_{i,2}, X_i\}_{i=1}^n$, we have
\begin{eqnarray*}
&&P_n\Big(\exists i\in\{1,\cdots,n\} : T_{i,1}<X_i<T_{i,2},~ T_{i,1}, T_{i,2}\in [t_0, t_0+B_0(n\log n)^{-1/3}] \Big)\\
&\leq&n\times (1+o(1))\frac{1}{6}f(t_0)h(t_0,t_0)B_0^3(n\log n)^{-1}=o(1), ~~a.s.
\end{eqnarray*}
Therefore, conditionally on $T$'s, the following equation holds with probability 1, almost surely,
\begin{eqnarray*}
&&n^{1/3}(\log n)^{-2/3}(W_{n,2}^*(t_0+ (n\log n)^{-1/3}t)-W_{n,2}^*(t_0))\\
&=&n^{1/3}(\log n)^{-2/3}\\
&&\times \biggr[
\int_{t_1\in[t_0,t_0+ (n\log n)^{-1/3}t]} \frac{\mathbf{1}_{x\leq t_1}}{F_n(t_1)^{2}}
+ \frac{\mathbf{1}_{t_1<x\leq t_2,  ~
t_2 > t_1+B_0 (n\log n)^{-1/3}}} {(F_n(t_2)-F_n(t_1))^{2}}d {\mathbb P}_n^*(x,t_1,t_2)\\
&&+\int_{t_2\in [t_0,t_0+ (n\log n)^{-1/3}t]} \frac{\mathbf{1}_{ t_1<x\leq t_2, ~ 
t_2>t_1+B_0 (n\log n)^{-1/3}}}{(F_n(t_2)-F_n(t_1))^{2}}+\frac{\mathbf{1}_{ x> t_2} }
{(1-F_n(t_2))^{2}}d {\mathbb P}_n^*(x,t_1,t_2)\biggr].
\end{eqnarray*}
The above display is equal to 
\begin{eqnarray*}
&&(1+o(1))n^{1/3}(\log n)^{-2/3} 
\\
&&\times \biggr[
\int_{t_1\in[t_0,t_0+ (n\log n)^{-1/3}t]} \frac{\mathbf{1}_{x\leq t_1}}{F_n(t_1)}+ \frac{\mathbf{1}_{t_1<x\leq t_2,  ~
t_2 > t_1+B_0 (n\log n)^{-1/3}}} {F_n(t_2)-F_n(t_1)}d H(t_1,t_2)\\
&&+\int_{t_2\in [t_0,t_0+ (n\log n)^{-1/3}t]}  \frac{\mathbf{1}_{ t_1<x\leq t_2, ~ 
t_2>t_1+B_0 (n\log n)^{-1/3}}}{F_n(t_2)-F_n(t_1)}+\frac{\mathbf{1}_{ x> t_2} }
{1-F_n(t_2)}dH(t_1,t_2)\biggr]\\
&=& (1+o(1))\int_{t_1\in [t_0,t_0+ (n\log n)^{-1/3}t]} \frac{\mathbf{1}_{t_1<x\leq t_2, ~ 
t_2>t_1+B_0 (n\log n)^{-1/3}}}{f(t_0)(t_2-t_1)}h(t_0,t_0)dt_1dt_2\\
&&+ (1+o(1))\int_{t_2\in [t_0,t_0+ (n\log n)^{-1/3}t]}\frac{\mathbf{1}_{t_1<x\leq t_2, ~ 
t_2>t_1+B_0 (n\log n)^{-1/3}}}{f(t_0)(t_2-t_1)}h(t_0,t_0)dt_1dt_2
+o(1)\\
&=& (1+o(1))\frac{2h(t_0,t_0)t}{3f(t_0)}.
\end{eqnarray*}
In the above display, recall that $H$ is the distribution function of $T_1$ and $T_2$.
This gives the first convergence result.  

The second convergence follows from a similar argument.  Conditionally on $T$'s, the following equation holds with probability 1, almost surely,
\begin{eqnarray*}
&&n^{2/3}(\log n)^{-1/3}\int_{t_0}^{t_0+(n\log n)^{1/3}t} (F_n(s)-F_n(t_0)) d W_{n,2}^*(s)\\
&=&n^{2/3}(\log n)^{-1/3} \biggr[
\int_{t_1\in[t_0,t_0+ (n\log n)^{-1/3}t],x\leq t_1}\frac{F_n(t_1)-F_n(t_0)}{F_n(t_1)^{2}}d{\mathbb P}^*_n(x,t_1,t_2)\\
&&+\int_{\scriptsize\begin{array}{c}
t_1\in [t_0,t_0+ (n\log n)^{-1/3}t], t_1<x\leq t_2, \\ 
t_2>t_1+B_0 (n\log n)^{-1/3}
\end{array}} 
\frac{F_n(t_1)-F_n(t_0)}{(F_n(t_2)-F_n(t_1))^{2}}d {\mathbb P}_n^*(x,t_1,t_2)\\
&&+\int_{\scriptsize\begin{array}{c} 
t_2\in [t_0,t_0+ (n\log n)^{-1/3}t], t_1<x\leq t_2, \\ 
t_2>t_1+B_0 (n\log n)^{-1/3}
\end{array}}
\frac{F_n(t_2)-F_n(t_0)} {(F_n(t_2)-F_n(t_1))^{2}}d {\mathbb P}_n^*(x,t_1,t_2)\\
&&+\int_{t_2\in [t_0,t_0+ (n\log n)^{-1/3}t], x> t_2}\frac{F_n(t_2)-F_n(t_0)}{(1-F_n(t_2))^{2}}d {\mathbb P}_n^*(x,t_1,t_2)\biggr].
\end{eqnarray*}
The preceding display is equal to 
\begin{eqnarray*}
&& (1+o(1))\int_{\scriptsize\begin{array}{c}
t_1\in [t_0,t_0+ (n\log n)^{-1/3}t], t_1<x\leq t_2, \\
t_2>t_1+B_0 (n\log n)^{-1/3}\end{array}}
 \frac{t_1-t_0}{t_2-t_1}h(t_0,t_0)dt_1dt_2\\
&&+ (1+o(1))\int_{\scriptsize\begin{array}{c}
t_2\in [t_0,t_0+ (n\log n)^{-1/3}t], t_1<x\leq t_2, \\ 
t_2>t_1+B_0 (n\log n)^{-1/3}\end{array}}
\frac{t_1-t_0}{t_2-t_1}h(t_0,t_0)dt_1dt_2
+o(1)\\
&=& (1+o(1))\frac{1}{3}h(t_0,t_0)t^2. 
\end{eqnarray*}
 
The third convergence result  follows from a similar argument as in the proof of Lemma 5.5 in \cite{groeneboom1991nonparametric}. For $t\in [0,B_0]$, let 
\begin{eqnarray*}
\bar W^*_{n}(t)
&=& n^{2/3}(\log n)^{-1/3}\biggr[
\int_{\scriptsize\begin{array}{c}
t_1\in[t_0,t_0+(n\log n)^{1/3}t],x\leq t_1,\\
t_2>t_1+B_0 (n\log n)^{-1/3}\end{array}}
 \frac{1}{F_n(t_1)}d{\mathbb P}^*_n(x,t_1,t_2)\\
&&-\int_{\scriptsize\begin{array}{c}
t_1\in[t_0,t_0+(n\log n)^{1/3}t], t_1<x\leq t_2\\
 t_2>t_1+B_0 (n\log n)^{-1/3}\end{array}}
\frac{1}{F_n(t_2)-F_n(t_1)}d {\mathbb P}_n^*(x,t_1,t_2)\\
&&+\int_{\scriptsize\begin{array}{c}
t_2\in[t_0,t_0+(n\log n)^{1/3}t], t_1<x\leq t_2 \\
 t_2>t_1+B_0 (n\log n)^{-1/3}\end{array}}
 \frac{1}{F_n(t_2)-F_n(t_1)}d {\mathbb P}_n^*(x,t_1,t_2)\\
&& -\int_{\scriptsize\begin{array}{c}
t_2\in[t_0,t_0+(n\log n)^{1/3}t], x> t_2, \\
 t_2>t_1+B_0 (n\log n)^{-1/3}\end{array}}
 \frac{1}{1-F_n(t_2)}d {\mathbb P}_n^*(x,t_1,t_2)\biggr].
\end{eqnarray*}
We can see that conditional on $T$'s,  $\bar W^*_{n}(t)$ is a martingale and its variance is given by 
\begin{eqnarray*}
&& (1+o(1))n^{1/3}(\log n)^{-2/3} \biggr[
\int_{t_1\in[t_0,t_0+ (n\log n)^{-1/3}t],x\leq t_1} \frac{1}{F_n(t_1)}dH(t_1,t_2)\\
&&+\int_{\scriptsize\begin{array}{c}
t_1\in [t_0,t_0+ (n\log n)^{-1/3}t], t_1<x\leq t_2, \\ 
t_2>t_1+B_0 (n\log n)^{-1/3}\end{array}}
\frac{1}{F_n(t_2)-F_n(t_1)}d H(t_1,t_2)\\
&&+\int_{\scriptsize\begin{array}{c}
t_2\in [t_0,t_0+ (n\log n)^{-1/3}t], t_1<x\leq t_2, \\ 
t_2>t_1+B_0 (n\log n)^{-1/3}\end{array}}
\frac{1}{F_n(t_2)-F_n(t_1)}d H(t_1,t_2)\\
&&+\int_{t_2\in [t_0,t_0+ (n\log n)^{-1/3}t], x> t_2}
 \frac{1}{1-F_n(t_2)}d H(t_1,t_2)\biggr]\\
 &=& (1+o(1))\frac{2h(t_0,t_0)t}{3f(t_0)},~~a.s.
\end{eqnarray*}
Therefore, by the martingale central limit theorem, we have the conditional weak convergence of $\bar W^*_{n}(t)$ to $ \sqrt{\frac{2}{3}h(t_0,t_0)/f(t_0)} \mathbb{Z}(t)$. Next we show that the difference between $\bar W^*_n$ and $n^{2/3}(\log n)^{-1/3}(W^*_{n,1}(t_0+(n\log n)^{1/3}t)-W^*_{n,1}(t_0))$ is ignorable.
Let $$\hat W^*_n= n^{2/3}(\log n)^{-1/3}(W^*_{n,1}(t_0+(n\log n)^{1/3}t)-W^*_{n,1}(t_0))-\bar W^*_n.$$
By Markov's inequality and a similar argument as in the proof of the second convergence result, we obtain that 
$$P_n \left( \max_{t\in[0,B_0]} |\hat W^*_n(t)|>\epsilon \right) \leq O(1)\frac{\epsilon^{-1}n^{2/3}}{(\log n)^{1/3}}\int_{t_1,t_2\in[t_0,t_0+(n\log n)^{-1/3}B_0]}dH(t_1,t_2)=o(1)~ a.s.$$
Therefore, we have the third convergence result. 
\end{proof}

\medskip
\begin{proof}[Proof of Theorem \ref{Thmcase2th}] We need the following tightness result, whose proof follows from a similar argument as in the proof of Lemma 5.6 in \cite{groeneboom1991nonparametric} and is omitted from this paper.
\begin{lemma}\label{lemma1}
If (\ref{condc2}) holds, then for any $\epsilon>0$ and $B_0>0$,  there exists a constant B such that the following result holds almost surely
$$P_n\left(\max_{x\in[-B_0,B_0]} (n\log n)^{1/3}\left|T^{*,(0)}_n(F_n(t_0)+(n\log n)^{-1/3}x)-t_0\right|>B\right)<\epsilon.$$
\end{lemma}
The above lemma and Lemma \ref{Thmcase2} imply that  $(T^{*,(0)}_n(F_n(t_0)+(n\log n)^{-1/3}x)-t_0$ converges to the process 
$$T(x):=\hbox{sargmin}_t\left\{\sqrt{\frac{2h(t_0,t_0)}{3f(t_0)}}\mathbb{Z}(t)+\frac{1}{3}h(t_0,t_0)t^2-x\frac{2h(t_0,t_0)t}{3f(t_0)}\right\}~~a.s.$$
Since  $P_n((n\log n)^{1/3} (F_n^{*,(1)}(t_0)-F_n(t_0))\leq x) = P_n(T^{*,(0)}_n(F_n(t_0)+(n\log n)^{-1/3}x)\geq t_0),$ we have 
$$P_n\left((n\log n)^{1/3} (F_n^{*,(1)}(t_0)-F_n(t_0))\leq x\right) \xrightarrow[]{a.s.} P(T(x)-f^{-1}(t_0)x\geq -f^{-1}(t_0)x).$$
By the stationary of process $T(x)-f^{-1}(t_0)x$ as given in \cite{groeneboom1989brownian}, we have that $P_n((n\log n)^{1/3} (F_n^{*,(1)}(t_0)-F_n(t_0))\leq x)$ converges to  $P(T(0)\geq -f^{-1}(t_0)x)$.
Then by a Brownian scaling argument as in the proof of Theorem~\ref{thm2}, we obtain the desired conclusion. 
\end{proof}

\bibliographystyle{mystyle}
\bibliography{mybib}
\end{document}